%% file: Bach_Breteaux_Tzaneteas-13-01-05.tex
\documentclass[12pt,a4paper,english]{article}
\usepackage[T1]{fontenc}
\usepackage[latin9]{inputenc}
\usepackage{xcolor}
\usepackage{pdfcolmk}
\usepackage{amsthm}
\usepackage{amsmath}
\usepackage{amssymb}
\usepackage{esint}
\PassOptionsToPackage{normalem}{ulem}
\usepackage{ulem}

\makeatletter

\pdfpageheight\paperheight
\pdfpagewidth\paperwidth

\providecolor{lyxadded}{rgb}{0,0,1}
\providecolor{lyxdeleted}{rgb}{1,0,0}

\theoremstyle{plain}
\newtheorem{thm}{Theorem}
\theoremstyle{definition}
  \newtheorem{rem}[thm]{Remark}
\theoremstyle{definition}
  \newtheorem{example}[thm]{Example}
\theoremstyle{definition}
  \newtheorem{defn}[thm]{Definition}
\theoremstyle{plain}
  \newtheorem{prop}[thm]{Proposition}
\theoremstyle{plain}
  \newtheorem{lem}[thm]{Lemma}

\usepackage{color}
\DeclareMathOperator{\Ex}{Ex}

\usepackage{color}
\usepackage{amssymb}%
\usepackage{amsmath}%
\usepackage{amsfonts}%
\usepackage{times}
\usepackage{amsthm} 
\usepackage{epsfig}

\newcommand{\vphi}{{\varphi}}           
\newcommand{\Om}{\Omega}                

\newcommand{\la}{\langle}                
\newcommand{\ra}{\rangle}                
\newcommand{\ol}{\overline}                
\newcommand{\one}{{\bf 1}}

\newcommand{\cB}{{\mathcal{B}}}

\newcommand{\cE}{{\mathcal{E}}}

\newcommand{\cL}{{\mathcal{L}}}         
\newcommand{\cO}{{\mathcal{O}}}         

\newcommand{\cS}{{\mathcal{S}}}

\newcommand{\cZ}{{\mathcal{Z}}}
\newcommand{\tcZ}{{\widetilde{\mathcal{Z}}}}

\newcommand{\fF}{{\mathfrak{F}}}

\newcommand{\fP}{{\mathfrak{P}}}         

\newcommand{\fS}{{\mathfrak{S}}}

\newcommand{\RR}{\mathbb{R}}            
\newcommand{\ZZ}{\mathbb{Z}}            
\newcommand{\NN}{\mathbb{N}}            
\newcommand{\CC}{\mathbb{C}}            
\newcommand{\UU}{\mathbb{U}}            
\newcommand{\tH}{\widetilde{H}}        
   
\newcommand{\talpha}{\tilde{\alpha}}
\newcommand{\tgamma}{\tilde{\gamma}}

\newcommand{\vAA}{{\vec{\mathbb{A}}}}            
\newcommand{\vG}{{\vec{G}}}             

\newcommand{\vnabla}{{\vec{\nabla}}}
\newcommand{\veps}{{\vec{\varepsilon}}}

\newcommand{\vO}{{\vec{0}}}

\newcommand{\vf}{{\vec{f}}}
\newcommand{\vk}{{\vec{k}}}

\newcommand{\vp}{{\vec{p}}}

\newcommand{\vx}{{\vec{x}}}

\newcommand{\sL}{{ }}
\newcommand{\rdm}{{1\!-\!\mathrm{pdm}}}
\newcommand{\coherent}{{\mathrm{coh}}}
\newcommand{\BHF}{{BHF}}
\newcommand{\tDM}{{\widetilde{\mathfrak{DM}}}}
\newcommand{\DM}{{\mathfrak{DM}}}
\newcommand{\QF}{{\mathfrak{QF}}}
\newcommand{\tpDM}{{\widetilde{\mathfrak{pDM}}}}
\newcommand{\pDM}{{\mathfrak{pDM}}}
\newcommand{\pQF}{{\mathfrak{pQF}}}
\newcommand{\cDM}{{\mathfrak{cDM}}}
\newcommand{\cQF}{{\mathfrak{cQF}}}

\newcommand{\dom}{{\rm dom}}
\newcommand{\rIm}{{\rm Im}}
\newcommand{\rRe}{{\rm Re}}               

\newcommand{\cirS}{\mathop{\bigcirc\kern -.73em {\scriptstyle{\rm S}}}}

\newcommand{\Tr}{{\rm Tr}}

\newcommand{\gs}{{\mathrm{gs}}}

\newcommand{\hf}{{H_f}}
\newcommand{\pf}{{\vec{P}_f}}
\newcommand{\nf}{N_f}

\renewcommand{\thesection}
{\Roman{section}}                      
\renewcommand{\theequation}
{\thesection.\arabic{equation}}        
\newcommand{\secct}[1]{\section{#1}
\setcounter{equation}{0}}              
\numberwithin{thm}{section}
 
\newtheorem{lemma}[thm]{Lemma}             
\theoremstyle{plain}
\def\beq{\begin{equation}}
\def\ene{\end{equation}}

\makeatother

\usepackage{babel}

\begin{document}

\title{Minimization of the Energy of the Non-Relativistic One-Electron
  Pauli-Fierz Model over Quasifree States}

\author{Volker Bach\footnote{Technische Universit\"at Braunschweig, Institut f\"ur Analysis und Algebra, Rebenring 31, A 14, 38106 Braunschweig, Germany.},
S\'ebastien Breteaux\footnotemark[1], Timmy Tzaneteas\footnote{Department of Mathematics, Ny Munkegade 118, building 1536, room 2.24, 8000 Aarhus C, Denmark.}
}
\maketitle
\begin{abstract}
In this article is proved the existence and uniqueness of a minimizer
of the energy for the non-relativistic one electron Pauli-Fierz model,
within the class of pure quasifree states. The minimum of the energy
on pure quasifree states coincides with the minimum of the energy on
quasifree states. Infrared and ultraviolet cutoffs are assumed, along
with sufficiently small coupling constant and momentum of the dressed
electron. A perturbative expression of the minimum of the energy on
quasifree states for a small momentum of the dressed electron and
small coupling constant is then given. We also express the Lagrange
equation for the minimizer, in terms of the generalized one particle
density matrix of the pure quasifree state.
\end{abstract}

\input{Intro_BBT-1-13-01-04.tex}
\begin{align} \label{eq-Ia.1}
\mathcal{E}_{g,\vec{p}} & (f,\gamma,\tilde{\alpha})\nonumber \\
 & =\frac{1}{2}\big\{(\Tr[\gamma\vec{k}]+f^{*}\vec{k}f+2\rRe(f^{*}\vec{G})-\vec{p})^{\cdot2}\nonumber \\
 & \phantom{=}+\Tr[\gamma\vec{k}\cdot\gamma\vec{k}]+\Tr[\tilde{\alpha}^{*}\vec{k}\cdot\tilde{\alpha}\vec{k}]+\Tr[|\vec{k}|^{2}\gamma]\nonumber \\
 & \phantom{=}+2\rRe(\overline{(\vec{G}+\vec{k}f)}^{*}\tilde{\alpha}(\vec{G}+\vec{k}f))+\Tr[(2\gamma+\mathbf{1})(\vec{G}+\vec{k}f)\cdot(\vec{G}+\vec{k}f)^{*}]\big\}
\nonumber \\
 & \phantom{=}+\Tr[\gamma|\vec{k}|]+f^{*}|\vec{k}|f\,.\end{align}
Now we are in position to formulate our main results.
\begin{thm}
\label{thm:inf_qf_inf_pqf} Let $0\leq\sigma<\Lambda<\infty$, $g\in\mathbb{R}$
and $\vec{p}\in\mathbb{R}^{3}$, $|\vec{p}|<1$. Minimizing the energy
over quasifree states is the same as minimizing the energy over pure
quasifree states, i.e.,
\begin{align} \label{eq-Ia.2}
E_{BHF}(g,\vec{p},\sigma,\Lambda) 
\ := \ 
\inf_{\rho\in\mathfrak{QF}}\Tr[H_{g,\vec{p}}\,\rho]
\ = \ 
\inf_{\rho\in\mathfrak{pQF}}\Tr[H_{g,\vec{p}}\,\rho]\,.
\end{align}
\end{thm}

\begin{thm}[Coherent States Case]
\label{thm:main_theorem_coherent} There exists a universal constant~$C<\infty$
such that, for $0\leq\sigma<\Lambda<\infty$, $g^{2}\ln(\Lambda+2)\leq C$
and $|\vec{p}|\leq1/3$, there exists a unique~$f_{g,\vec{p}}$
which minimizes~$\mathcal{E}_{g,\vec{p}}(f)=\mathcal{E}_{g,\vec{p}}(f,0,0)$
in~$L^{2}(S_{\sigma,\Lambda}\times\mathbb{Z}_{2},(\tfrac{1}{2}|\vec{k}|^{2}+|\vec{k}|)dk)$.
\begin{enumerate}
\item \label{enu:closed-system-f-u-1-1}The minimizer $f_{g,\vec{p}}$ solves
the system of equations
$$\begin{cases}
f_{g,\vec{p}} & =\frac{\vec{u}_{g,\vec{p}}\cdot\vec{G}}{\frac{1}{2}|\vec{k}|^{2}+|\vec{k}|-\vec{k}\cdot\vec{u}_{g,\vec{p}}}\,,\\
\vec{u}_{g,\vec{p}} & =\vec{p}-2\rRe(f_{g,\vec{p}}^{*}\vec{G})-f_{g,\vec{p}}^{*}\vec{k}f_{g,\vec{p}}\,,\end{cases}$$
with $|\vec{u}_{g,\vec{p}}|\leq|\vec{p}|$.
\item For~$0\leq\sigma<\Lambda<\infty$,\[
\inf_{f\in L^{2}(S_{\sigma,\Lambda}\times\mathbb{Z}_{2})}\mathcal{E}_{g,\vec{p}}\,(f)=\inf_{f\in L^{2}(S_{\sigma,\Lambda}\times\mathbb{Z}_{2},(\tfrac{1}{2}|\vec{k}|^{2}+|\vec{k}|)dk)}\mathcal{E}_{g,\vec{p}}\,(f)=\mathcal{E}_{g,\vec{p}}\,(f_{g,\vec{p}})\,,\]
and for~$0<\sigma<\Lambda<\infty$, $f_{g,\vec{p}}\in L^{2}(S_{\sigma,\Lambda}\times\mathbb{Z}_{2})$.
\item \label{enu:asymptotic-devel-energy-1-1} For fixed~$g$, $\sigma$, $\Lambda$, and small values of $|\vec{p}|$, we have that \[
\mathcal{E}_{g,\vec{p}}(f_{g,\vec{p}})=\mathcal{E}_{g,\vec{p}}(0)-\vec{p}\cdot\vec{G}^{*}\frac{1}{\frac{1}{2}|\vec{k}|^{2}+|\vec{k}|+2\vec{G}\cdot\vec{G}^{*}}\vec{G}\cdot\vec{p}+\mathcal{O}(|\vec{p}|^{3})\,.\]

\end{enumerate}
\end{thm}

We summarize in Theorem~\ref{thm:main_thm_QF} the information obtained
in Sections~\ref{sec:Existence-and-uniqueness} to~\ref{sec:Lagrange-equations}.
\begin{thm}[Quasifree States Case]
\label{thm:main_thm_QF}Let~$0<\sigma<\Lambda<\infty$. There exists
$C>0$ (possibly depending on $\sigma$ and $\Lambda$) such that for all $|g|,|\vec{p}|<C$, there exists a unique $(f_{g,\vec{p}},\gamma_{g,\vec{p}},\tilde{\alpha}_{g,\vec{p}})$
which minimizes the energy $\mathcal{E}_{g,\vec{p}}(f,\gamma,\tilde{\alpha})$.
\begin{enumerate}
\item The dependence of $(f_{g,\vec{p}},\gamma_{g,\vec{p}},\tilde{\alpha}_{g,\vec{p}})$
on~$(g,\vec{p})$ is smooth.
\item The functions $(f_{g,\vec{p}},\gamma_{g,\vec{p}},\tilde{\alpha}_{g,\vec{p}})$
satisfy\begin{align*}
f_{g,\vec{p}} & =\big(\frac{1}{2}|\vec{k}|^{2}+|\vec{k}|\big)^{-1}\vec{p}.\vec{G}+\mathcal{O}(\|(g,\vec{p})\|^{3})\,,\\
\tilde{\alpha}_{g,\vec{p}} & =-\tilde{S}^{-1}(\vec{G}\cdot\vec{G}^{*})+\mathcal{O}(\|(g,\vec{p})\|^{3})\,,\\
\gamma_{g,\vec{p}}+\gamma_{g,\vec{p}}^{2} & =\tilde{\alpha}_{g,\vec{p}}\,\tilde{\alpha}_{g,\vec{p}}^{*}\,,\end{align*}
\textup{where $\tilde{S}$ acts on the kernel $K_{A}(k,k^{\prime})$ of a Hilbert-Schmidt operator~$A$ as the multiplication by 
$\vec{k}\cdot\vec{k}^{\prime}+\frac{1}{2}|\vec{k}|^{2}+|\vec{k}|+\frac{1}{2}|\vec{k}^{\prime}|^{2}+|\vec{k}^{\prime}|$.}
\item For fixed~$\sigma$, $\Lambda$, and small values of $|g|$ and $|\vec{p}|$, we have that\begin{multline*}
E_{BHF}(g,\vec{p},\sigma,\Lambda)\\
=\mathcal{E}_{g,\vec{p}}(0,0,0)-g^{2}|\vec{p}|^{2}C_{2,2}(\sigma,\Lambda)-g^{4}C_{4,0}(\sigma,\Lambda)+\mathcal{O}(\|(g,\vec{p})\|^{5})\,,\end{multline*}
as~$(g,\vec{p})\to0$, with $C_{2,2}(\sigma,\Lambda)=(2\pi^{2}-\frac{8}{3}\pi)\ln(\frac{\Lambda+2}{\sigma+2})$
and $C_{4,0}(\sigma,\Lambda)>0$.
\item The minimizer $(f_{g,\vec{p}},\gamma_{g,\vec{p}},\tilde{\alpha}_{g,\vec{p}})$
satisfies (we drop the $g, \vec{p}$ indexes to simplify the notation)\begin{align*}
M(\gamma,\vec{u})f & =-(\vec{k}(\gamma+\frac{1}{2}\mathbf{1})-\vec{u})\cdot\vec{G}-\vec{k}\cdot\tilde{\alpha}(\vec{G}+\vec{k}f)\,,\\
\mathcal{A}(\lambda)\tilde{\alpha} & =-(\vec{G}+\vec{k}f)\cdot(\vec{G}+\vec{k}f)^{*}\,,\\
\gamma+\gamma^{2} & =\tilde{\alpha}\,\tilde{\alpha}^{*}\,,\\
\lambda: & =\int_{0}^{\infty}\negthickspace\negthickspace e^{-t(\frac{1}{2}+\gamma)}(M(\gamma,\vec{u})+(\vec{G}+\vec{k}f)\!\cdot\!(\vec{G}+\vec{k}f)^{*})e^{-t(\frac{1}{2}+\gamma)}dt\,,\\
\vec{u}: & =\vec{p}-\Tr[\gamma\vec{k}]-f^{*}\vec{k}f-2\rRe(f^{*}\vec{G})\,,\end{align*}
with 
\begin{align*}
M(\gamma,\vec{u}) \ := \ & \frac{1}{2}|\vec{k}|^{2}+|\vec{k}|-\vec{k}\cdot\vec{u}+\vec{k}\cdot\gamma\vec{k}, \\
\mathcal{A}(\lambda)\tilde{\alpha} \ := \ & \vec{k}\tilde{\alpha}\cdot\vec{k}+\lambda\tilde{\alpha}+\tilde{\alpha}\lambda .
\end{align*}
\end{enumerate}
\end{thm}

\begin{rem}
In the coherent states case the formula\[
\mathcal{E}_{g,\vec{p}}(f_{g,\vec{p}})=\mathcal{E}_{g,\vec{p}}(0)-g^{2}|\vec{p}|^{2}C_{2,2}(\sigma,\Lambda)+\mathcal{O}(\|(g,\vec{p})\|^{5})\,,\]
holds and can easily be compared to the quasifree state case.
\end{rem}

\begin{rem}
Although Theorem~\ref{thm:main_thm_QF} is formulated in terms of
the one-particle reduced density matrix~$\Gamma_{\rho}$ and its
constituents~$\gamma_{\rho}$ and $\tilde{\alpha}_{\rho}$, it turns
out to be more convenient to parametrize the pureness constraint $\gamma_{\rho}+\gamma_{\rho}^{2}=\tilde{\alpha}_{\rho}\tilde{\alpha}_{\rho}^{*}$
in terms of an antilinear Hilbert-Schmidt operator $\hat{r}$ which
is chosen such that $\gamma_{\rho}=\frac{1}{2}(\cosh(2\hat{r})-\mathbf{1})$,
$\tilde{\alpha}_{\rho}=\frac{1}{2}\sinh(2\hat{r})\mathcal{J}$, where
$\mathcal{J}:f\in L^{2}(S_{\sigma,\Lambda}\times\mathbb{Z}_{2})\mapsto\bar{f}\in L^{2}(S_{\sigma,\Lambda}\times\mathbb{Z}_{2})$.
This is explained in detail in Section~\ref{sec:Quasifree-states}.
\end{rem}

\paragraph{Outline of the article}

We introduce our notation to describe the second quantization
framework in
Section~\ref{sec:The-second-quantization}. Section~\ref{sec:Quasifree-states}
introduces two parametrizations of pure quasifree states and contains
the proof of Theorem~\ref{thm:inf_qf_inf_pqf}. The energy functional
for a fixed value of the momentum $\vec{p}$ of the dressed electron is
computed in Section~\ref{sec:Energy-functional}, along with some
positivity properties of the different parts of the energy. From
Section~\ref{sec:Minimization-coherent-states} on we tacitly assume
that the coupling constant $|g| >0$ is small. The energy is then
minimized in the particular case of coherent states in
Section~\ref{sec:Minimization-coherent-states}, providing a first
upper bound to the energy of the ground state and a proof of
Theorem~\ref{thm:main_theorem_coherent}.  The existence and uniqueness
of a minimizer among the class of pure quasifree state is then proven
in Section~\ref{sec:Existence-and-uniqueness} provided $|\vec{p}|$ is
small enough. The first terms of a perturbative expansion for small
$g$ and $\vec{p}$ of the energy at the minimizer is computed in
Section~\ref{sec:Perturbative-approach}. Finally the Lagrange
equations associated with the problem of minimization in the
generalized one particle density matrix variables are presented in
Section~\ref{sec:Lagrange-equations}.

\section{\label{sec:The-second-quantization}Second Quantization}

In this section $\mathcal{Z}$ denotes a $\mathbb{C}$-Hilbert space
with a scalar product $\mathbb{C}$-linear in the right variable and
$\mathbb{C}$-antilinear in the left variable.

Let $\mathcal{B}(X;Y)$ be the space of bounded operators between
two Banach spaces~$X$ and~$Y$, and $\mathcal{L}^{1}(\mathcal{Z})$
the space of trace class operators on~$\mathcal{Z}$. Given two $\mathbb{C}$-Hilbert
spaces $(\mathcal{Z}_{j},\langle\cdot,\cdot\rangle_{j})$, $j=1,2$
and a bounded linear operator $A:\mathcal{Z}_{1}\to\mathcal{Z}_{2}$,
set $A^{*}:\mathcal{Z}_{2}\to\mathcal{Z}_{1}$ to be the operator
such that\[
\forall z_{1}\in\mathcal{Z}_{1},\, z_{2}\in\mathcal{Z}_{2},\qquad\langle z_{2},Az_{1}\rangle_{2}=\overline{\langle z_{1},A^{*}z_{2}\rangle_{1}}\,,\]
and $\rRe A:=\frac{1}{2}(A\oplus A^{*})$, $\rIm A:=\frac{1}{2i}(A\oplus(-A^{*}))\in\mathcal{B}(\mathcal{Z}_{1},\mathcal{Z}_{2})\oplus\mathcal{B}(\mathcal{Z}_{2},\mathcal{Z}_{1})$.
\begin{example}
For $z$, $z^{\prime}\in\mathcal{Z}$, \[
\langle z,z^{\prime}\rangle=z^{*}z^{\prime}\,.\]

The adjoint of a bounded operator $A$ on $\mathcal{Z}$ is $A^{*}$.
\end{example}
The symmetrization operator $\mathcal{S}_{n}$ on $\mathcal{Z}^{\otimes n}$
is the orthogonal projection defined by\[
\mathcal{S}_{n}(z_{1}\otimes\cdots\otimes z_{n})=\frac{1}{n!}\sum_{\pi\in\mathfrak{S}_{n}}z_{\pi_{1}}\otimes\cdots\otimes z_{\pi_{n}} \]
and extension by linearity and continuity.
The symmetric tensor product for vectors is $z_{1}\vee z_{2}=\mathcal{S}_{n_{1}+n_{2}}(z_{1}\otimes z_{2})$
and more generally for operators is $A_{1}\vee A_{2}=\mathcal{S}_{q_{1}+q_{2}}\circ(A_{1}\otimes A_{2})\circ\mathcal{S}_{p_{1}+p_{2}}$
for $A_{j}\in\mathcal{B}(\mathcal{Z}^{\otimes p_{j}};\mathcal{Z}^{\otimes q_{j}})$.
We set 
\begin{align*}
\mathcal{Z}^{\vee n} \ := \ \mathcal{S}_{n}\mathcal{Z}^{\otimes n}, \qquad
\mathcal{B}^{p,q} \ := \ \mathcal{B}(\mathcal{Z}^{\otimes p};\mathcal{Z}^{\otimes q}).
\end{align*}
\begin{defn}
The symmetric Fock space on a Hilbert space $\mathcal{Z}$ is defined
to be \[
\mathfrak{F}_{+}(\mathcal{Z}):=\bigoplus_{n=0}^{\infty}\mathcal{Z}^{\vee n}\,,\]
where $\mathcal{Z}^{\vee 0} := \CC \Omega$, $\Omega$ being the normalized vacuum vector.

For a linear operator~$C$ on~$\mathcal{Z}$ such that $\|C\|_{\mathcal{B}(\mathcal{Z})}\leq1$,
let $\Gamma(C)$ defined on each~$\mathcal{Z}^{\vee n}$ by~$C^{\vee n}$
and extended by continuity to the symmetric Fock space on~$\mathcal{Z}$.

For an operator~$X$ on~$\mathcal{Z}$, the second quantization~$d\Gamma(X)$
of~$X$ is defined on each~$\mathcal{Z}^{\vee n}$ by\[
d\Gamma(X)\Big|_{\mathcal{Z}^{\vee n}}=n\,\mathbf{1}_{\mathcal{Z}}^{\vee n-1}\vee X\]
 and extended by linearity to $\bigoplus_{n\geq0}^{alg}\mathcal{Z}^{\vee n}$.
The number operator is~$N_{f}=d\Gamma(\mathbf{1}_{\mathcal{Z}})$.
For a vector $f$ in $\mathcal{Z}$, the creation and annihilation
operators in~$f$ are the linear operators such that $a(f)\Omega = 0$, $a^*(f)\Omega = f$, and 
\begin{align} \label{eq-VB.1}
a(f)g^{\vee n}=\sqrt{n}(f^{*}g)\; g^{\vee n-1}\,, \ \ \mbox{and} \ \ a^{*}(f)g^{\vee n}=\sqrt{n+1}f\vee g^{\vee n}\,,
\end{align}
for all $g\in\mathcal{Z}$. By the polarization identity
\[
\forall g_{1},\dots,g_{n},\qquad g_{1}\vee\cdots\vee g_{n}=\frac{1}{2^{n}n!}\sum_{\varepsilon_{i}=\pm1}\varepsilon_{1}\cdots\varepsilon_{n}\Big(\sum_{j=1}^{n}\varepsilon_{j}g_{j}\Big)^{\otimes n}
\]
Eq.~\eqref{eq-VB.1} extends to $\mathcal{Z}^{\vee n}$ and hence also to $\bigoplus_{n\geq0}^{alg}\mathcal{Z}^{\vee n}$. They satisfy the
canonical commutation relations $[a(f),a^{*}(g)]=f^{*}g$, $[a(f),a(g)]=[a^{*}(f),a^{*}(g)]=0$.

The self-adjoint field operator associated to $f$ is~$\Phi(f)= \tfrac{1}{\sqrt{2}}\big( a^{*}(f) + a(f) \big)$.
For more details on the second quantization see the book of Berezin~\cite{MR0208930}.

\end{defn}
A dot ``$\cdot$'' denotes an operation analogous to the scalar product
in~$\mathbb{R}^{3}$. For every two objects $\vec{a}=(a_{1},a_{2},a_{3})$
and $\vec{b}=(b_{1},b_{2},b_{3})$ with three components such that
the products $a_{j}b_{j}$ are well defined \[
\vec{a}\cdot\vec{b}:=\sum_{j=1}^{3}a_{j}b_{j}\,.\]
%
\begin{example}
With $\vec{p}\in\mathbb{R}^{3}$, $\vec{G}\in\mathcal{Z}^{3}$, $\vec{k}\in(\mathcal{B}^{1,1})^{3}$\begin{alignat*}{5}
\vec{p}^{\,\cdot2} & =\sum_{j=1}^{3}p_{j}^{2}\in\mathbb{R}\,, & \vec{k}\cdot\vec{p} & =\sum_{j=1}^{3}p_{j}k_{j}\in\mathcal{B}^{1,1}, & \vec{p}\cdot\vec{G} & =\sum_{j=1}^{3}p_{j}G_{j}\in\mathcal{Z}\,,\\
\vec{k}^{\,\cdot2} & =\sum_{j=1}^{3}k_{j}^{2}\in\mathcal{B}^{1,1}, & \vec{k}\cdot\vec{G} & =\sum_{j=1}^{3}k_{j}G_{j}\in\mathcal{Z}\,, & \; \vec{G}^{*}\cdot\vec{k} & =\sum_{j=1}^{3}G_{j}^{*}k_{j}\in\mathcal{Z}^{*},\\
\vec{G}\cdot\vec{G}^{*} & =\sum_{j=1}^{3}G_{j}G_{j}^{*}\in\mathcal{B}^{1,1}, &\; \vec{G}^{*}\!\cdot\vec{G} & =\sum_{j=1}^{3}G_{j}^{*}G_{j}\in\mathbb{C}\,,\end{alignat*}
where for an object with three components $\vec{a}=(a_{1},a_{2},a_{3})$
such that $a_{j}^{*}$ is well-defined, $\vec{a}^{*} := (a_{1}^{*},a_{2}^{*},a_{3}^{*})$.
We sometimes use the notation $\vec{p}^{\,\cdot2}=|\vec{p}|^{2}$,
or $\vec{k}^{\,\cdot2}=|\vec{k}|^{2}$.

And with another product, such as the symmetric tensor product $\vee$,\begin{align*}
\vec{k}^{\,\cdot\vee2} & =\sum_{j=1}^{3}k_{j}^{\vee2}\in\mathcal{B}^{2,2}\,, & \vec{k}\cdot\!\vee\vec{G} & =\sum_{j=1}^{3}k_{j}\vee G_{j}\in\mathcal{B}^{2,3}\,.\end{align*}
\end{example}
Recall that the Weyl operators are the unitary operators $W(f)=\exp(i\Phi(f))$ 
satisfying the relations \begin{align}
\forall z_{1},\, z_{2}\in\mathcal{Z}: \qquad W(z_{1})W(z_{2}) & =e^{-\frac{i}{2}\rIm(z_{1}^{*}z_{2})}W(z_{1}+z_{1})\,,\label{eq:product_Weyl_operators}\\
\forall z\in\mathcal{Z}: \qquad W(-i\sqrt{2}z)\Omega & =e^{-\frac{|z|^{2}}{2}}\sum_{n=0}^{\infty}\frac{z^{\vee n}}{\sqrt{n!}}\,.\label{eq:coherent_state_series}\end{align}
\begin{defn}
\label{def:coherent-states}
The \emph{coherent vectors} are the vectors of the form\[
E_{z}=W(-i\sqrt{2}z)\Omega\]
for some $z\in\mathcal{Z}$ and the \emph{coherent states} are the
states of the form\[
E_{z}E_{z}^{*}\,.\]
\end{defn}
\begin{defn}
A \emph{symplectomorphism}~$T$ for the symplectic form~$\rIm\langle\cdot,\cdot\rangle$
on a $\mathbb{C}$-Hilbert space~$\mathcal{Z}$ is a continuous $\mathbb{R}$-linear
automorphism on~$\mathcal{Z}$ which preserves this symplectic form, i.e.,
\[
\forall z_{1},\, z_{2}\in\mathcal{Z}: \qquad\rIm\langle Tz_{1},Tz_{2}\rangle=\rIm\langle z_{1},z_{2}\rangle\,.\]

A symplectomorphism $T$ is \emph{implementable} if there is a unitary
operator~$\mathbb{U}_{T}$ on~$\fF_+(\mathcal{Z})$ such that \[
\forall z\in\mathcal{Z}\,,\quad\mathbb{U}_{T}W(z)\mathbb{U}_{T}^{*}=W(Tz)\,.\]
In this case $\mathbb{U}_{T}$ is a \emph{Bogolubov transformation}
corresponding to~$T$.
\end{defn}
We recall a well-known parametrization, in the spirit of the polar
decomposition, of implementable symplectomorphisms.
\begin{prop}
\label{pro:polar-decomposition}The set of implementable symplectomorphisms
is the set of operators\[
T=u \exp[\hat{r}]=u\;\sum_{n=0}^{\infty}\frac{1}{n!}\hat{r}^{n}\,,\]
where~$u$ is an isometry and~$\hat{r}$ is an antilinear operator,
self-adjoint in the sense that $\forall z,z^{\prime}\in\mathcal{Z},\,\langle z,\hat{r}z^{\prime}\rangle=\langle z^{\prime},\hat{r}z\rangle$,
and Hilbert-Schmidt in the sense that the positive operator~$\hat{r}^{2}$
is trace-class. Equivalently, there exist a Hilbert basis $(\varphi_{j})_{j\in\mathbb{N}}$
of $\mathcal{Z}$ and $(\hat{r}_{i,j})_{i,j}\in\ell^{2}(\mathbb{N}^{2};\mathbb{C})$
such that\[
\hat{r}=\sum_{i,j=1}^{\infty}\hat{r}_{i,j}\,\langle\cdot,\varphi_{j}\rangle\,\varphi_{i}\qquad\mbox{and}\qquad\forall i,j\in\mathbb{N}^{2}:\quad\hat{r}_{i,j}=\hat{r}_{j,i}\,.\]

\end{prop}
\begin{proof}
On the one hand, every operator of the form $T=u \exp[\hat{r}]$ with
$u$ a unitary operator and $\hat{r}$ a self-adjoint antilinear operator
is a symplectomorphism. Since a unitary operator is a symplectomorphism,
and the set of symplectomorphisms is a group for the composition,
it is enough to prove that $\exp[\hat{r}]$ is a symplectomorphism.
It is indeed the case since, for all $z$, $z'$ in $\mathcal{Z}$,\begin{align*}
\rIm\langle e^{\hat{r}}z,e^{\hat{r}}z'\rangle & =\rIm\langle e^{\hat{r}}z,\cosh(\hat{r})z'\rangle+\rIm\langle e^{\hat{r}}z,\sinh(\hat{r})z'\rangle\\
 & =\rIm\langle\cosh(\hat{r})e^{\hat{r}}z,z'\rangle+\rIm\langle z',\sinh(\hat{r})e^{\hat{r}}z\rangle\\
 & =\rIm\langle\cosh(\hat{r})e^{\hat{r}}z,z'\rangle-\rIm\langle\sinh(\hat{r})e^{\hat{r}}z,z'\rangle\\
 & =\rIm\langle e^{-\hat{r}}e^{\hat{r}}z,z'\rangle\,.\end{align*}
The implementability condition is then satisfied if we suppose $\hat{r}$
to be Hilbert-Schmidt.

On the other hand, to get exactly this formulation we give the step
to go from the result given in Appendix~A in~\cite{BRETEAUX2012}
to the decomposition in Proposition~\ref{pro:polar-decomposition}.
In~\cite{BRETEAUX2012} an implementable symplectomorphism is decomposed as
\begin{equation}
T=ue^{c\tilde{r}}\,,\label{eq:vielle-decomp-polaire}\end{equation}
where~$u$ is a unitary operator, $c$ is a conjugation and~$\tilde{r}$
is a Hilbert-Schmidt, self-adjoint, non-negative operator commuting
with~$c$. It is then enough to set~$\hat{r}=c\tilde{r}$ to get
the expected decomposition. To check the self-adjointness of~$\hat{r}$,
observe that, for all~$z$,~$z'$ in~$\mathcal{Z}$, \[
\langle z',\hat{r}z\rangle=\langle z',\tilde{r}cz\rangle=\langle\tilde{r}z',cz\rangle=\langle z,c\tilde{r}z'\rangle=\langle z,\hat{r}z'\rangle\,.\]

For the convenience of the reader we recall the main steps to obtain
the decomposition in Eq.~(\ref{eq:vielle-decomp-polaire}).
First decompose $T$ in its $\mathbb{C}$-linear and antilinear parts,
$T=L+A$, then write the polar decomposition~$L=u|L|$. It is then
enough to prove that $|L|+u^{*}A$ is of the form $e^{c\tilde{r}}$.
From certain properties of symplectomorphisms (also recalled in~\cite{BRETEAUX2012})
it follows that the antilinear operator $u^{*}A$ is selfadjoint
and $|L|^{2}+\mathbf{1}_{\mathcal{Z}}=(u^{*}A)^{2}$. A decomposition
of the positive trace class operator $(u^{*}A)^{2}=\sum_{j}\lambda_{j}^{2}e_{j}e_{j}^{*}$
with~$e_{j}$ an orthonormal basis of~$\mathcal{Z}$ yields $|L|=\sum_{j}(1+\lambda_{j}^{2})^{1/2}e_{j}e_{j}^{*}$.
Using that $\lambda_{j}\to0$ one can study the operator $|L|$ and
$u^{*}A$ on the finite dimensional subspaces~$\ker(|L|-\mu\mathbf{1}_{\mathcal{Z}})$
which are invariant under~$u^{*}A$. It is then enough to prove that
for a $\mathbb{C}$-antilinear self-adjoint operator $f$ such that~$ff^{*}=\lambda^{2}$
on a finite dimensional space, there is an orthonormal basis $\{\varphi_{k}\}_k$
such that $f(\varphi_{k})=\lambda \varphi_{k}$. The conjugation is then defined
such that~$c(\sum\beta_{k}\varphi_{k})=\sum\bar{\beta}_{k}\varphi_{k}$ and $\tilde{r}=\sinh^{-1}(\lambda_{j})\mathbf{1}$
on that subspace.
\end{proof}

\section{\label{sec:Quasifree-states}Pure Quasifree States}

\subsection{From Quasifree States to Pure Quasifree States}

Let $\mathcal{Z}$ be the $\mathbb{C}$-Hilbert space $L^2(S_{\sigma,\Lambda}\times\mathbb{Z}_2)$. We make use of the following characterization of quasifree density matrices.
\begin{lem}
The set of quasifree density matrices and pure quasifree density matrices, respectively, 
of finite photon number expectation value can be characterized by\begin{align*}
\mathfrak{QF}=\mathfrak{DM}\bigcap\Big\{ W(-i & \sqrt{2}f)\mathbb{U}^{*}\frac{\Gamma(C)}{\Tr[\Gamma(C)]}\mathbb{U}W(-i\sqrt{2}f)^{*}\\
\Big|\, & f\in\mathcal{Z},\,\mathbb{U}\mbox{ a Bogolubov transformation},\\
 & C\in\mathcal{L}^{1}(\mathcal{Z}),\, C\geq0,\,\|C\|_{\mathcal{B}(\mathcal{Z})}<1\Big\}\end{align*}
\begin{align*}
\mathfrak{pQF}=\mathfrak{DM}\bigcap\Big\{ W(-i & \sqrt{2}f)\mathbb{U}^{*}\Omega\Omega^{*}\mathbb{U}W(-i\sqrt{2}f)^{*}\\
\Big|\, & f\in\mathcal{Z},\,\mathbb{U}\mbox{ a Bogolubov transformation}\Big\}\end{align*}
\end{lem}
\begin{proof} We only sketch the argument, details can be found in \cite{MR0208930, MR2238912}.
It is not difficult to see that any density matrix of the form $W(-i \sqrt{2}f)\mathbb{U}^{*} \frac{\Gamma(C)}{\Tr[\Gamma(C)]} \mathbb{U}W(-i\sqrt{2}f)^{*}$
is indeed quasifree. Conversely, if $\rho \in \mathfrak{QF}$ is a quasifree density matrix then it is fully characterized by its
one-point function $f_\rho \in \cZ$ and two-point functions $(\gamma_\rho, \tilde{\alpha}_\rho)$.
Moreover, $W(-i\sqrt{2}f_\rho)^* \, \rho \, W(-i\sqrt{2}f_\rho) \in \mathfrak{cQF}$ is a centered quasifree density matrix with the same
one-particle density matrix, that is, the density matrix $W(-i\sqrt{2}f_\rho)^* \, \rho \, W(-i\sqrt{2}f_\rho)$ corresponds to 
$(0, \gamma_\rho - f_\rho f_\rho^*, \talpha_\rho - f_\rho \bar f_\rho^*)$.
Obviously, $\gamma_\rho - f_\rho f_\rho^*$ is again trace-class and $\talpha_\rho - f_\rho \bar f_\rho^*$ is Hilbert-Schmidt.
Now, we use that there exists a Bogolubov transformation $\mathbb{U}$ which eliminates $\talpha_\rho$, i.e., 
$\mathbb{U}^* W(\sqrt{2}f_\rho/i)^* \, \rho \, W(\sqrt{2}f_\rho/i) \mathbb{U}$ corresponds to $(0, \tgamma_\rho, 0)$. While this is the
only nontrivial step of the proof, we note that if $\mathbb{U}$ is characterized by $u$ and $v$ as in Lemma~\ref{lemma-2} then
there is an involved, but explicit formula that determines $u$ and $v$. Again $\tgamma_\rho$ is
trace-class because the photon number operator $\nf$ transforms under $\mathbb{U}^*$ to itself plus lower order corrections,
$\mathbb{U}^* \, \nf \, \mathbb{U} = \nf + \cO(\nf^{1/2} + 1)$. Finally, it is easy to see that $(0, \tgamma_\rho, 0)$
corresponds to the quasifree density matrix $\Gamma(C_\rho)/\Tr[\Gamma(C_\rho)]$ with $C_\rho := \tgamma_\rho (1 + \tgamma_\rho)^{-1}$.
Following these steps we finally obtain $$\rho = W(f_\rho) \mathbb{U} \frac{\Gamma(C_\rho)}{\Tr[\Gamma(C_\rho)]} \mathbb{U}^* W(f_\rho)^{*}\,,$$
as asserted. The additional characterization of pure quasifree density matrices is obvious.
\end{proof}
\begin{lem} \label{lemma-2}
Let $\mathbb{U}\in\mathcal{B}(\mathfrak{F})$ be a unitary operator. 
The following statements are equivalent:
\begin{align} \label{eq-VB-lemma-2.1}
& \text{$\mathbb{U}\in\mathcal{B}(\mathfrak{F})$ is a Bogolubov transformation;} 
\\[1ex] \label{eq-VB-lemma-2.4}
\Leftrightarrow & \text{$\exists T$ implementable symplectomorphism,}\\
& \qquad\qquad\qquad\mathbb{U}=\tilde{U}_{T}\,, \quad \tilde{U}_{T}W(f)\tilde{U}^{*}=W(Tf). \nonumber
\\[1ex] \label{eq-VB-lemma-2.2}
\Leftrightarrow & \exists u \in\mathcal{B}(\mathcal{Z}), \; v\in\mathcal{L}^{2}(\mathcal{Z}) \ 
\forall f\in\mathcal{Z}: \\
&\qquad\qquad\qquad \mathbb{U}a^{*}(f)\mathbb{U}^{*}=a^{*}(uf)+a(\mathcal{J}v\mathcal{J}f);\nonumber
\\[1ex] \label{eq-VB-lemma-2.3}
\Leftrightarrow & \mathbb{U}=\exp(iH), \ \text{where $H = H^*$ is a semibounded operator,}\\ 
&\qquad\qquad\qquad\text{quadratic in $a^*$ and $a$ and without linear term.}\nonumber
\end{align}
\end{lem}
\begin{proof} Again, we only sketch the argument. First note that \eqref{eq-VB-lemma-2.1}$\Leftrightarrow$\eqref{eq-VB-lemma-2.4} is
the definition of a Bogolubov transformation. Secondly, $\tilde{U}_{T}W(f)\tilde{U}_{T}^{*}=W(Tf)$ is equivalent to
$\tilde{U}_{T} \Phi(f) \tilde{U}_{T}^{*} = \Phi(Tf)$. Hence, using that $a^*(f) = \frac{1}{\sqrt{2}}[\Phi(f) - i \Phi(if)]$ and
$a(f) = \frac{1}{\sqrt{2}}[\Phi(f) + i \Phi(if)]$ we obtain the equivalence \eqref{eq-VB-lemma-2.4}$\Leftrightarrow$\eqref{eq-VB-lemma-2.2}.
Thirdly, setting $U_\lambda =\exp(i\lambda H)$ and $a_\lambda^*(f) := U_\lambda a^*(f) U_\lambda^*$, 
we observe that $\partial_\lambda a_\lambda^*(f) = i[H, a_\lambda^*(f)]$. Furthermore, $[H, a_\lambda^*(f)]$ is linear in $a^*$ and $a$ if, and only if,
$H$ is quadratic in $a^*$ and $a$. Solving this linear differential equation, we finally obtain \eqref{eq-VB-lemma-2.2}$\Leftrightarrow$\eqref{eq-VB-lemma-2.3}.
\end{proof}
\begin{lem}
For all Bogolubov transformation $\mathbb{U}$ and $g\in\mathcal{Z}$:
\begin{align}
W(g)\mathbb{U}\,\mathfrak{QF}\,\mathbb{U}^{*}W(g)^{*} \ = \ & \mathfrak{QF}, \\
\mathbb{U}\,\mathfrak{cQF}\,\mathbb{U}^{*} \ = \ & \mathfrak{cQF}.
\end{align}
\end{lem}
\begin{rem}
A pure quasifree state is a particular case of quasifree state with
$C=0$, that is~$\Gamma(C)=\Omega\Omega^{*}$.
\end{rem}
We come to the main result of this section.
\begin{thm}
\label{thm:inf_qf_egal_inf_pqf}Let $0\leq\sigma<\Lambda<\infty$,
$g\in\mathbb{R}$ and $\vec{p}\in\mathbb{R}^{3}$, $|\vec{p}|<1$.
Minimizing the energy over quasifree states is the same as minimizing
the energy over pure quasifree states, i.e.,\[
E_{BHF}(g,\vec{p},\sigma,\Lambda) 
\ := \ 
\inf_{\rho\in\mathfrak{QF}}\Tr[H_{g,\vec{p}}\,\rho]
\ = \ 
\inf_{\rho\in\mathfrak{pQF}}\Tr[H_{g,\vec{p}}\,\rho] 
\,.\]
\end{thm}
For the proof of Theorem~\ref{thm:inf_qf_egal_inf_pqf} we derive a couple of preparatory lemmata.
\begin{prop}
Let~$C$ a non-negative operator on~$\mathcal{Z}$, then\[
\Big\{ \Tr[\Gamma(C)]<\infty \Big\} 
\qquad \Leftrightarrow \qquad 
\Big\{ C\in\mathcal{L}^{1}(\mathcal{Z})\quad\mbox{and}\quad\|C\|_{\mathcal{B}(\mathcal{Z})}<1 \Big\}\,.\]
In this case~$\Tr[\Gamma(C)]=\det(1-C)^{-1}$. (We refrain from defining the determinant.)
For the direction $\Leftarrow$ the non-negativity assumption is not
necessary.\end{prop}
\begin{proof}
Let us decompose~$\mathcal{Z}=\bigoplus_{j\geq0}\mathbb{C}e_{j}$
where $C=\sum c_{j}e_{j}e_{j}^{*}$ with~$(e_{j})_{j\geq0}$ an orthonormal
basis of~$\mathcal{Z}$. Then~$\mathfrak{F}_{+}(\mathcal{Z})=\bigotimes_{j\geq0}\mathfrak{F}_{+}(\mathbb{C}e_{j})$
and\[
\Tr[\Gamma(C)]=\Tr[\bigotimes_{j\geq0}\Gamma(c_{j})]=\prod_{j\geq0}\Tr[\Gamma(c_{j})]=\prod_{j\geq0}\frac{1}{1-c_{j}}\]
and the infinite product converges exactly when~$C\in\mathcal{L}^{1}(\mathcal{Z})$
and~$\|C\|_{\mathcal{B}(\mathcal{Z})}<1$.\end{proof}
\begin{lem}
\label{lem:Gamma(C)_as_integral}Suppose $\mathcal{Z}_{d}$ is of
dimension~$d<\infty$. Then, for any non-negative operator~$C_{d}\neq0$
such that~$C_{d}\in\mathcal{L}^{1}(\mathcal{Z}_{d})$ and~$\|C_{d}\|_{\mathcal{B}(\mathcal{Z}_{d})}<1$,
there exist a non-negative measure~$\mu_{d}$ (depending on~$C$)
of mass one on~$\mathcal{Z}_{d}$ and a family $\{ \rho_{d}(z_{d}) \}_{z_{d} \in \mathcal{Z}_{d}}$ of pure quasifree
states such that\[
\frac{\Gamma(C)}{\Tr[\Gamma(C)]}=\int_{\mathcal{Z}_{d}}\rho_{d}(z_{d})\, d\mu_{d}(z_{d})\,.\]
\end{lem}
\begin{proof}
In finite dimension $d$ we can use a resolution of the identity with
coherent states (see, e.g., \cite{MR0208930})\[
\mathbf{1}_{\Gamma(\mathcal{Z}_{d})}=\int_{\mathcal{Z}_{d}}E_{z_{d}}E_{z_{d}}^{*}\,\frac{dz_{d}}{\pi^{d}}\]
where~$\mathcal{Z}_{d}$ is identified with~$\mathbb{C}^{d}$ and~$dz_{d} = dx_{d} \, dy_{d}$,
$z_{d}=x_{d}+iy_{d}$. Using Equation~(\ref{eq:coherent_state_series})
we get\begin{align*}
\Gamma(C) & =\int_{\mathcal{Z}_{d}}\Gamma(C^{1/2})E_{z_{d}}E_{z_{d}}^{*}\Gamma(C^{1/2})\,\frac{dz_{d}}{\pi^{d}}\\
 & =\int_{\mathcal{Z}_{d}}E_{C^{1/2}z_{d}}E_{C^{1/2}z_{d}}^{*}\,\frac{\exp(|C^{1/2}z_{d}|^{2}-|z_{d}|^{2})dz_{d}}{\pi^{d}}\,.\end{align*}
The measure~$d\mu_{d}(z_{d})=\pi^{-d}\exp(|C^{1/2}z_{d}|^{2}-|z_{d}|^{2})dz_{d}/\Tr[\Gamma(C)]$
has mass one. Indeed\begin{align*}
\int_{\mathcal{Z}_{d}}\exp(-z_{d}^{*}(\mathbf{1}_{\mathcal{Z}_{d}}-C)z_{d})\frac{dz_{d}}{\pi^{d}} & =\prod_{j=1}^{d}\int_{\mathbb{R}^{2}}\exp(-(1-c_{j})(x^{2}+y^{2}))\frac{dx\, dy}{\pi}\\
 & =\prod_{j=1}^{d}\frac{1}{1-c_{j}}=\Tr[\Gamma(C)]\end{align*}
where~$C=\sum_{j=1}^{d}c_{j}e_{j}e_{j}^{*}$ with $(e_{j})_{j=1}^{d}$
an orthonormal basis of~$\mathcal{Z}_{d}$.\end{proof}

\begin{proof}[Proof of Theorem~\ref{thm:inf_qf_egal_inf_pqf}]
The inclusion $\mathfrak{pQF}\subset\mathfrak{QF}$ implies that 
\[
\inf_{\rho\in\mathfrak{QF}}\Tr[H_{g,\vec{p}}\,\rho]
\ \leq \ 
\inf_{\rho\in\mathfrak{pQF}}\Tr[H_{g,\vec{p}}\,\rho] \, ,
\]
and it is hence enough to prove for any quasifree
state \[
\rho_{qf} \ = \ W(-i\sqrt{2}f)\,\mathbb{U}_{T}^{*}\,\frac{\Gamma(C)}{\Tr[\Gamma(C)]}\,\mathbb{U}_{T}\, W(-i\sqrt{2}f)^{*}\,,\]
that the inequality\[
\Tr[H_{g,\vec{p}}\,\rho_{qf}] \ \geq \ \inf_{\rho\in\mathfrak{pQF}} \Tr[H_{g,\vec{p}}\,\rho]
\,.\]
holds true. The operator~$C$ is decomposed as~$C=\sum_{j\geq0}c_{j}e_{j}e_{j}^{*}$
where~$(e_{j})$ is an orthonormal basis of the Hilbert space~$\mathcal{Z}$
and~$c_{j}\geq0$. Let~$C_{d}=\sum_{j\leq d}c_{j}e_{j}e_{j}^{*}$.
Let\[
\rho_{qf,d}=W(-i\sqrt{2}f)\,\mathbb{U}_{T}^{*}\,\frac{\Gamma(C_{d})}{\Tr[\Gamma(C_{d})]}\,\mathbb{U}_{T}\, W(-i\sqrt{2}f)^{*}\,,\]
then using Lemma~\ref{lem:Gamma(C)_as_integral} with $\mathcal{Z}_{d}=\bigoplus_{j\leq d}\mathbb{C}e_{j}$,
$\mathfrak{F}_{+}\mathcal{Z}=\mathfrak{F}_{+}(\mathcal{Z}_{d}\oplus\mathcal{Z}_{d}^{\perp})\cong\mathfrak{F}_{+}\mathcal{Z}_{d}\otimes\mathfrak{F}_{+}\mathcal{Z}_{d}^{\perp}$
and the extension of the operator $\Gamma(C_{d})$ on~$\mathfrak{F}_{+}\mathcal{Z}_{d}$
to~$\mathfrak{F}_{+}\mathcal{Z}_{d}\otimes\mathfrak{F}_{+}\mathcal{Z}_{d}^{\perp}$
by~$\Gamma(C_{d})\otimes(\Omega_{\mathcal{Z}_{d}^{\perp}}\Omega_{\mathcal{Z}_{d}^{\perp}}^{*})$
(which we still denote by~$\Gamma(C_{d})$), we obtain \[
\rho_{qf,d} \ = \ \int_{\mathcal{Z}_{d}} \rho_{d}(z_{d}) \: d\mu_{d}(z_{d})
\,,\]
where~$\rho_{d}(z_{d})$ are pure quasifree states and the~$\mu_{d}$
are non-negative measures with mass one. Note that\[
\nu_{d} \ := \ \frac{\Tr[\Gamma(C_{d})]}{\Tr[\Gamma(C)]}=\prod_{j>d}(1-c_{j}) 
\ \nearrow \ 1 
\, , \]
as $d\to\infty$. Further note that $\rho_{qf} \geq \nu_d \, \rho_{qf,d}$, for any $d \in \NN$, since $\Gamma(C) \geq \Gamma(C_d)$. Thus
\begin{align*}
\Tr[H_{g,\vec{p}}\,\rho_{qf}] 
& \ \geq \ 
\Tr[H_{g,\vec{p}}\,\nu_{d}\rho_{qf,d}]
\\ & \ = \ 
\nu_{d} \, \int_{\mathcal{Z}_{d}} \Tr[H_{g,\vec{p}}\,\rho_{d}(z_{d})]\: d\mu_{d}(z_{d})
\\ & \ \geq \ 
\nu_{d} \, \inf_{z_{d}\in\mathcal{Z}_{d}}\Tr[H_{g,\vec{p}}\,\rho_{d}(z_{d})]
\\ & \ \geq \
\nu_{d} \, \inf_{\rho\in\mathfrak{pQF}}\Tr[H_{g,\vec{p}}\,\rho]\,,
\end{align*}
for all $d \in \NN$, and in the limit $d \to \infty$, we obtain 
\[ \Tr[H_{g,\vec{p}}\,\rho_{qf}]  
\ \geq \ 
\lim_{d \to \infty}\{ \nu_{d} \} \, \inf_{\rho\in\mathfrak{pQF}}\Tr[H_{g,\vec{p}}\,\rho]
\ = \ 
\inf_{\rho\in\mathfrak{pQF}}\Tr[H_{g,\vec{p}}\,\rho]
\, . \qedhere\]
\end{proof}

\subsection{Pure Quasifree States and their One-Particle Density Matrices}

Let $\mathcal{Z}$ be a $\mathbb{C}$-Hilbert space.
\begin{defn}
Let $\rho \in \mathfrak{DM}$ be a density matrix on the bosonic Fock space $\mathfrak{F}_{+}(\mathcal{Z})$ over $\mathcal{Z}$.
If $\Tr[\rho N_{f}^{\frac{p+q}{2}}]<\infty$, we define $\rho^{p,q}\in\mathcal{B}^{p,q}(\mathcal{Z})$
through\[
\forall\varphi,\psi\in\mathcal{Z}\,,\qquad\psi^{*\vee p}\rho^{p,q}\varphi^{\vee q}=\Tr[a^{*}(\varphi)^{q}a(\psi)^{p}\rho]\,.\]
We single out 
\[
f=\rho^{0,1}\in\mathcal{B}^{0,1}\cong\mathcal{Z} \, ,
\]
i.e., $f_\rho \in \cZ$ is the unique vector such that $\Tr[a(\psi) \, \rho] = \psi^* f_\rho$, for all $\psi \in \cZ$. Furthermore,
with $\tilde{\rho}=W(\sqrt{2}f_{\rho}/i)^{*}\rho W(\sqrt{2}f_{\rho}/i)$,
the matrix elements of the (generalized) one-particle density matrix are defined by \[
\gamma_{\rho}=\tilde{\rho}^{1,1}\in\mathcal{B}^{1,1}\qquad\mbox{and}\qquad\alpha_{\rho}=\tilde{\rho}^{0,2}\in\mathcal{B}^{0,2}\cong\mathcal{Z}^{\vee2}\,,\]
in other words\begin{align*}
\forall\varphi,\psi\in\mathcal{Z}\,: \qquad\langle\psi\,,\gamma_{\rho}\,\varphi\rangle & =\Tr[\tilde{\rho}\, a^{*}(\varphi)a(\psi)]\,,\\
\langle\psi\otimes\varphi\,,\alpha_{\rho}\rangle & =\Tr[\tilde{\rho}\, a(\psi)a(\varphi)]\,.\end{align*}
Note that $f_{\rho}$, $\gamma_{\rho}$, and $\alpha_{\rho}$ exist for any $\rho \in \mathfrak{DM}$ since
$\nf \rho, \rho \nf \in \cL^1(\fF_+)$.
\end{defn}
\begin{rem}
For a centered pure quasifree state $\tilde{\rho}$, $\tilde{\rho}^{p,q}$
vanishes when $p+q$ is odd.
\end{rem}
\begin{rem}
Another definition of the one-particle density matrix~$\gamma_{\rho}$
would be through the relation~$\langle\psi,\gamma_{\rho}\varphi\rangle=\Tr[a^{*}(\varphi)a(\psi)\rho]$.
We prefer here a definition with a {}``centered'' version~$\tilde{\rho}$
of the state~$\rho$, because this centered quasifree state~$\tilde{\rho}$
then satisfies the usual Wick theorem. The same considerations hold
for~$\alpha_{\rho}$.

Hence, any quasifree density matrix is characterized by~$(f_{\rho},\gamma_{\rho},\alpha_{\rho})$, since
$\rho^{p,q}$ can be expressed in terms of $(f_{\rho},\gamma_{\rho},\alpha_{\rho})$.

When $f_{\rho}=0$, the definition of $\gamma_{\rho}$ is consistent
with the usual one, for $z_{1}$, $z_{2}\in\mathcal{Z}$, $\langle z_{1},\gamma_{\rho}z_{2}\rangle=\Tr[a^{*}(z_{2})a(z_{1})\rho]$. 
The definition of $\alpha_{\rho}$ is related with the definition
of the operator $\hat{\alpha}_{\rho}$ (here denoted with a hat for
clarity) used in the article of Bach, Lieb and Solovej~\cite{MR1297873},
through the relation $\langle z_{1}\otimes z_{2},\alpha_{\rho}\rangle_{\mathcal{Z}^{\otimes2}}=\langle z_{1},\tilde{\alpha}_{\rho}cz_{2}\rangle_{\mathcal{Z}}$
with $c$ a conjugation on $\mathcal{Z}$.\end{rem}
\begin{example}
A centered pure quasifree state satisfies the relation, \begin{equation}
\tilde{\rho}^{2,2}=\gamma\otimes\gamma+\gamma\otimes\gamma\,\Ex+\alpha\alpha^{*}\in\mathcal{B}^{2,2}\,,\label{eq:rho_2_2_quasi_free}\end{equation}
where the exchange operator is the linear operator on $\mathcal{Z}^{\otimes2}$
such that\[
\forall z_{1},\, z_{2}\in\mathcal{Z},\qquad\Ex(z_{1}\otimes z_{2})=z_{2}\otimes z_{1}\]
and where for any $b\in\mathcal{Z}^{\otimes2}$, $\alpha\alpha^{*}b=\langle\alpha,b\rangle_{\mathcal{Z}^{\otimes2}}\,\alpha$.\end{example}
We now turn to another parametrization of quasifree states, by vectors
in a real Hilbert space. This parametrization enables us to use convexity
arguments.
\begin{prop}
\label{pro:gamma_r_alpha_r}Let $T=ue^{\hat{r}}$ be an implementable
symplectomorphism and $\rho$ a quasifree state of the form $\rho=\mathbb{U}_{T}^{*}\Omega\Omega^{*}\mathbb{U}_{T}$.
Then 
\begin{align} \label{eq:gamma_alpha_as_ch_sh-1}
\gamma_{\rho} \ = \ & \tfrac{1}{2}(\cosh(2\hat{r})-\mathbf{1}) \, ,
\\ \label{eq:gamma_alpha_as_ch_sh-2}
\forall z_{1},z_{2}\in\mathcal{Z}: \quad \;
\langle z_{1}\otimes z_{2},\alpha_{\rho}\rangle_{\mathcal{Z}^{\otimes2}} 
\ = \ & \langle z_{1},\tfrac{1}{2}\sinh(2\hat{r})z_{2}\rangle \, .
\end{align}
\end{prop}
\begin{proof}[Proof of Proposition~\ref{pro:gamma_r_alpha_r}]
We have $Ti=ue^{\hat{r}}i=uie^{-\hat{r}}=iue^{-\hat{r}}$ and for
all $z\in\mathcal{Z}$ 
\begin{align*}
\Tr[\rho W(-i\sqrt{2}z)] & = \Tr\big[\mathbb{U}_{T}^{*}\Omega\Omega^{*}\mathbb{U}_{T}W(-i\sqrt{2}z)\big]\\
 & = \Omega^{*}W(ue^{\hat{r}}(-i\sqrt{2}z))\Omega\\
 & = \Omega^{*}W(-i\sqrt{2}ue^{-\hat{r}}z)\Omega\\
 & = \exp\big(-\tfrac{1}{2}|ue^{-\hat{r}}z|^{2} \big) \\
 & = \exp\big(-\tfrac{1}{2}|e^{-\hat{r}}z|^{2} \big)
\end{align*}
From this formula we can easily compute the function
\begin{align*}
h(t,s) \ := \ 
\Tr\big[\rho W(-ti\sqrt{2}z)W(-si\sqrt{2}z)\big]
\ = \ 
\exp\big(-\tfrac{1}{2}|e^{-\hat{r}}(t+s)z|^{2}\big)
\end{align*}
whose derivative $\partial_{t}\partial_{s}$ at $(t,s) = (0,0)$ involves
$\alpha$ and $\gamma$
\begin{align*}
\partial_{t}\partial_{s}h(0,0) & = \Tr\big[\rho(a^{*}(z)-a(z))^{2}\big] \\
 & = -2z^{*}\gamma z+2\rRe(\alpha^{*}z^{\vee2})-z^{*}z\,.
\end{align*}
But we have also
\begin{align*}
\partial_{t}\partial_{s}\exp & (-\frac{1}{2}|e^{-\hat{r}}(t+s)z|^{2})\Big|_{t=s=0}\\
 & =-(e^{-\hat{r}}z)^{*}(e^{-\hat{r}}z)\\
 & =-(\cosh(\hat{r})z-\sinh(\hat{r})z)^{*}(\cosh(\hat{r})z-\sinh(\hat{r})z)\\
 & =-(\cosh(\hat{r})z)^{*}(\cosh(\hat{r})z)\\
& \qquad+2\rRe(\sinh(\hat{r})z)^{*}(\cosh(\hat{r})z)-(\sinh(\hat{r})z)^{*}(\sinh(\hat{r})z)\\
 & =-z^{*}(\cosh^{2}\hat{r}+\sinh^{2}\hat{r})z+2\rRe(z^{*}(\sinh\hat{r}\cosh\hat{r})z)\\
 & =-z^{*}\cosh(2\hat{r})z+2\rRe(z^{*}\frac{1}{2}\sinh(2\hat{r})z)
\end{align*}
and hence, using the polarization identity\[
4z\vee z' \ = \ (z+z')^{\otimes2}-(z-z')^{\otimes2} \]
to recover every vector from $\mathcal{Z}^{\vee2}$ from linear combinations
of vectors of the form $z^{\vee2}$, we arrive at \eqref{eq:gamma_alpha_as_ch_sh-1}-\eqref{eq:gamma_alpha_as_ch_sh-2}.
\end{proof}
\begin{prop}
The admissible $\gamma$, $\alpha$ for a pure quasifree state are
exactly those satisfying the relation\begin{equation}
\gamma+\gamma^{2}=(\alpha\otimes\mathbf{1})^{*}(\mathbf{1}\otimes\alpha)\,,\label{eq:constraint}\end{equation}
with $\gamma\geq0$.
\end{prop}
This is the constraint when we minimize the energy as a function of
$(f,\gamma,\alpha)$ with the method of Lagrange multipliers in Section~\ref{sec:Lagrange-equations}.
\begin{proof}
If $\gamma$, $\alpha$ are associated with a quasifree state, then
there is an $\hat{r}$ such that $\gamma$, $\alpha$ and $\hat{r}$
satisfy Equations~(\ref{eq:gamma_alpha_as_ch_sh-1}) and~(\ref{eq:gamma_alpha_as_ch_sh-2}), then\begin{align*}
\langle z_{1},(\alpha^{*}\otimes\mathbf{1})(\mathbf{1}\otimes\alpha)z_{2}\rangle & =(\alpha^{*}\otimes z_{1}^{*})(z_{2}\otimes\alpha)\\
 & =([\alpha^{*}(z_{2}\otimes\mathbf{1})]\otimes z_{1}^{*})\alpha\\
 & =\big\langle\alpha^{*}(z_{2}\otimes\mathbf{1}),\tfrac{1}{2}\sinh(2\hat{r})z_{1}\big\rangle_{\mathcal{Z}}\\
 & =\big\langle\alpha^{*},z_{2}\otimes\tfrac{1}{2}\sinh(2\hat{r})z_{1}\big\rangle_{\mathcal{Z}^{\otimes2}}\\
 & =\big\langle\tfrac{1}{4}\sinh^{2}(2\hat{r})z_{1},z_{2}\big\rangle_{\mathcal{Z}}\\
 & =\big\langle(\tfrac{1}{2}(\cosh(2\hat{r})-\mathbf{1})+\tfrac{1}{4}(\cosh(2\hat{r})-\mathbf{1})^{2})z_{1},z_{2}\big\rangle_{\mathcal{Z}}\,.
\end{align*}
Conversely, if $\gamma$ and $\alpha$ satisfy Eq.~(\ref{eq:constraint})
then we define the $\mathbb{C}$-antilinear operator $\hat{\alpha}$
such that $\langle z_{1},\hat{\alpha}z_{2}\rangle=(z_{1}\otimes z_{2})^{*}\alpha$,
and set $\hat{r}=\frac{1}{2}\sinh^{-1}(2\hat{\alpha})$, then\[
\forall z_{1},z_{2}\in\mathcal{Z}: \quad
\langle z_{1}\otimes z_{2}, \alpha_{\rho} \rangle_{\mathcal{Z}^{2}} 
\ = \ 
\langle z_{1}, \hat{\alpha} z_{2} \rangle
\ = \
\big\langle z_{1}, \tfrac{1}{2}\sinh(2\hat{r})z_{2} \big\rangle \, , 
\]
which, in turn, implies that $(\alpha^{*}\otimes\mathbf{1})(\mathbf{1}\otimes\alpha)=\frac{1}{4}\sinh^{2}(2\hat{r})$.
Hence, we have
\[
\gamma+\gamma^{2}=\frac{1}{4}\sinh^{2}(2\hat{r})
\]
and as $\gamma \geq 0$, it follows that $\gamma=\frac{1}{2}(\cosh(2\hat{r})-\mathbf{1})$.
Then $\gamma$, $\alpha$ is associated with the centered pure quasifree
state whose symplectic transformation is $\exp[\hat{r}]$.

\end{proof}

\section{\label{sec:Energy-functional}Energy Functional}

\emph{Notation:} We first recall that, as before, we denote by~$\vec{k}$, and $|\vec{k}|$
the multiplication operators $\vec{k}\otimes\mathbf{1}_{\mathbb{C}^{2}}$ and $|\vec{k}|\otimes\mathbf{1}_{\mathbb{C}^{2}}$
on $\mathcal{Z}=L^{2}(S_{\sigma,\Lambda}\times\mathbb{Z}_{2})$, with
three components in the case of $\vec{k}$.

We now work at fixed values of total momentum $\vec{p} \in \RR^3$. The operator $H_{g,\vec{p}}$
is given by 
\[
H_{g,\vec{p}}=\frac{1}{2}(d\Gamma(\vec{k})+2\rRe \: a^{*}(\vec{G})-\vec{p})^{\cdot2}+d\Gamma(|\vec{k}|)
\,,\]
where $\vec{G}(k)=\vec{G}(\vec{k},\pm):= g |\vec{k}|^{-1/2} \vec{\varepsilon}_{\pm}(\vec{k})$.
The energy of a pure quasifree state $\rho$ associated with $f \in \cZ$, $\gamma \in \cL^1(\cZ)$, $\alpha \in \cZ^{\vee 2}$ is
\begin{align} \label{eq-VB-3}
\mathcal{E}_{g,\vec{p}}(f,\gamma,\alpha) 
\ := \ 
\Tr[H_{g,\vec{p}}\,\rho]
\,,
\end{align}
where $\cZ$ is the $\mathbb{C}$-Hilbert space $\mathcal{Z}=L^{2}(S_{\sigma,\Lambda}\times\mathbb{Z}_{2})$
and $\cL^1(\cZ)$ is the space of trace class operators on $\mathcal{Z}$.
\begin{prop}
The energy functional \eqref{eq-VB-3} is
\begin{align}
\mathcal{E}_{g,\vec{p}}(f,\gamma,\alpha) & = \frac{1}{2}\Big\{(\Tr[\gamma\vec{k}]+f^{*}\vec{k}f+2\rRe(f^{*}\vec{G})-\vec{p})^{\cdot2}\nonumber \\
 & \phantom{=}+\Tr[\gamma\vec{k}\cdot\gamma\vec{k}]+\alpha^{*}(\vec{k}\cdot\otimes\vec{k})\alpha+\Tr[|\vec{k}|^{2}\gamma]\nonumber \\
 & \phantom{=}+2\rRe\{\alpha^{*}[(\vec{G}+\vec{k}f)^{\cdot\vee2}]\}+\Tr[(2\gamma+\mathbf{1})(\vec{G}+\vec{k}f)\cdot(\vec{G}+\vec{k}f)^{*}]\Big\}\nonumber \\
 & \phantom{=}+\Tr[\gamma|\vec{k}|]+f^{*}|\vec{k}|f\,.\label{eq:Energy}
\end{align}
where the following positivity properties hold
\begin{align*}
(\Tr[\gamma\vec{k}]+f^{*}\vec{k}f+2\rRe(f^{*}\vec{G})-\vec{p})^{\cdot2} & \geq0\,,\\
\Tr[\gamma\vec{k}\cdot\gamma\vec{k}]+\Tr[\gamma\vec{k}]^{\cdot2}+\alpha^{*}(\vec{k}\cdot\otimes\vec{k})\alpha+\Tr[|\vec{k}|^{2}\gamma] & \geq0\,,\\
(\Tr[\gamma\vec{k}]+f^{*}\vec{k}f+2\rRe(f^{*}\vec{G})-\vec{p})^{\cdot2}\qquad\\
+\Tr[\gamma\vec{k}\cdot\gamma\vec{k}]+\alpha^{*}(\vec{k}\cdot\otimes\vec{k})\alpha+\Tr[|\vec{k}|^{2}\gamma] & \geq0\,,\\
2\rRe(\alpha^{*}((\vec{G}+\vec{k}f)^{\cdot\vee2}))+\Tr[(2\gamma+\mathbf{1})(\vec{G}+\vec{k}f)\cdot(\vec{G}+\vec{k}f)^{*}] & \geq0\,.
\end{align*}
The energy of a pure quasifree state in the variables $f$ and $\hat{r}$ is
\begin{align}
\hat{\mathcal{E}}_{g,\vec{p}}(f,\hat{r}) & =\tfrac{1}{2}\big\{(\Tr[\tfrac{1}{2}(\cosh(2\hat{r})-1)\vec{k}]+f^{*}\vec{k}f+2\rRe(f^{*}\vec{G})-\vec{p})^{\cdot2}\nonumber \\
 & \phantom{=}+\Tr[\tfrac{1}{2}(\cosh(2\hat{r})-1)\vec{k}\cdot\tfrac{1}{2}(\cosh(2\hat{r})-1)\vec{k}]\nonumber \\
 & \phantom{=}+\Tr[\tfrac{1}{2}\sinh(2\hat{r})\vec{k}\cdot\tfrac{1}{2}\sinh(2r)\vec{k}]+\Tr[|\vec{k}|^{2}\tfrac{1}{2}(\cosh(2\hat{r})-1)]\nonumber \\
 & \phantom{=}+2\rRe\langle\tfrac{1}{2}\sinh(2\hat{r})(\vec{G}+\vec{k}f);(\vec{G}+\vec{k}f)\rangle\nonumber \\
 & \phantom{=}+\Tr[(2\tfrac{1}{2}(\cosh(2\hat{r})-1)+\mathbf{1})(\vec{G}+\vec{k}f)\cdot(\vec{G}+\vec{k}f)^{*}]\big\}\nonumber \\
 & \phantom{=}+\Tr[\tfrac{1}{2}(\cosh(2\hat{r})-1)|\vec{k}|]+f^{*}|\vec{k}|f\,.\end{align}
\end{prop}
\begin{proof}
Using the Weyl operators,
\[
\mathcal{E}_{g,\vec{p}}(f,\gamma,\alpha):=\Tr[H_{g,\vec{p}}\rho]=\Tr[H_{g,\vec{p}}(f)\tilde{\rho}]\]
where~$H_{g,\vec{p}}(f)=W(\sqrt{2}f/i)^{*}H_{g,\vec{p}}W(\sqrt{2}f/i)$
and $\tilde{\rho}=W(\sqrt{2}f/i)^{*}\rho W(\sqrt{2}f/i)$, so that
$\tilde{\rho}$ is centered. Modulo terms of odd order, which
vanish when we take the trace against a centered quasifree state, $H_{g,\vec{p}}(f)$ equals
\begin{align*}
H_{g,\vec{p}}(f) & = \frac{1}{2}\big(d\Gamma(\vec{k})+f^{*}\vec{k}f+2\rRe\big(a^{*}(\vec{k}f+\vec{G})\big)+2\rRe(f^{*}\vec{G})-\vec{p}\big)^{\cdot2}\\
 & \qquad+d\Gamma(|\vec{k}|)+f^{*}|\vec{k}|f + \mathrm{odd}\\
 & = \frac{1}{2}\big(d\Gamma(\vec{k})+f^{*}\vec{k}f+2\rRe(f^{*}\vec{G})-\vec{p}\big)^{\cdot2}\\
 & \qquad+\frac{1}{2}\big(2\rRe\big(a^{*}(\vec{k}f+\vec{G})\big)\big)^{\cdot2}+d\Gamma(|\vec{k}|)+f^{*}|\vec{k}|f + \mathrm{odd}\,.
\end{align*}
To compute $\mathcal{E}(f,\gamma,\alpha)$ we are thus lead to compute,
for $\vec{\varphi}\in\mathcal{Z}^{3}$ and $\vec{u}\in\mathbb{R}^{3}$,
\[
\Tr\big[ \tilde{\rho} \, (d\Gamma(\vec{k})+\vec{u})^{\cdot2} \big]
\qquad \mbox{and} \qquad
\Tr\big[ \tilde{\rho} \, (2\rRe\{a(\vec{\varphi})\})^{\cdot2} \big]
\,.\]
The expression of the energy as a function of $(f,\gamma,\alpha)$
then follows from Propositions~\ref{pro:positivity1} and~\ref{pro:positivity2}.
The expression of the energy as a function of $(f,r)$ follows from
Proposition~\ref{pro:gamma_r_alpha_r}.\end{proof}
\begin{prop}
\label{pro:positivity1}Let $\vec{u}\in\mathbb{R}^{3}$, then\begin{align*}
0\leq\Tr[\tilde{\rho}(d\Gamma(\vec{k})+\vec{u})^{\cdot2}] & =(\Tr[\gamma\vec{k}]+\vec{u})^{\cdot2}-\Tr[\gamma\vec{k}]^{\cdot2}\\
 & \phantom{=}+\Tr[\gamma\vec{k}\cdot\gamma\vec{k}]+\Tr[\gamma\vec{k}]^{\cdot2}+\alpha^{*}(\vec{k}\cdot\!\otimes\vec{k})\alpha+\Tr[|\vec{k}|^{2}\gamma]\,.\end{align*}
\end{prop}
This condition is used with $\vec{u}=\vec{p}-f^{*}\vec{k}f-2\rRe(f^{*}\vec{G})$.
\begin{proof}
Indeed, \[
(d\Gamma(\vec{k})+\vec{u})^{\cdot2} \ = \ d\Gamma(\vec{k})^{\cdot2}+2d\Gamma(\vec{k})\cdot\vec{u}+\vec{u}{}^{\,\cdot2} \, . 
\]
Then we use that $\Tr[\tilde{\rho}\, d\Gamma(\vec{k})]=\Tr[\gamma\vec{k}]$,
add and substract $\Tr[\gamma\vec{k}]^{\cdot2}$ to complete the square and
compute $\Tr[\tilde{\rho}\, d\Gamma(\vec{k})^{\cdot2}]$ using Lemma~\ref{lem:Tr[rho_XX]}.
\end{proof}
\begin{lem}
\label{lem:Tr[rho_XX]}Let $X\in\mathcal{B}^{1,1}$, then 
\[
0 \ \leq \ \Tr\big[\tilde{\rho}d\Gamma(X)d\Gamma(X)^{*}\big] \ = \ 
\Tr[\gamma X\gamma X^{*}]+|\Tr[\gamma X]|^{2}+\alpha^{*}(X\otimes X^{*})\alpha+\Tr[XX^{*}\gamma]
\,.\]
\end{lem}
\begin{proof}
Indeed, using Equation~(\ref{eq:rho_2_2_quasi_free}),\begin{align*}
\Tr & [\tilde{\rho}d\Gamma(X)d\Gamma(X)^{*}]\\
 & =\Tr[\tilde{\rho}(\negmedspace\int\!\! X\!(k_{1},k_{1}^{\prime})X\!(k_{2},k_{2}^{\prime})a^{*}\!(k_{1})a^{*}\!(k_{2})a(k_{2}^{\prime})a(k_{1}^{\prime})dk_{1}dk_{2}dk_{1}^{\prime}dk_{2}^{\prime}+d\Gamma(XX^{*})]\\
 & =\Tr[(\gamma\otimes\gamma+\gamma\otimes\gamma\, Ex+\alpha\alpha^{*})(X\otimes X^{*})]+\Tr[\gamma\, XX^{*}]\\
 & =\Tr[\gamma X]\Tr[\gamma X^{*}]+\Tr[\gamma X\gamma X^{*}]+\alpha^{*}(X\otimes X^{*})\alpha+\Tr[\gamma\, XX^{*}]\,.\tag*{\qedhere}\end{align*}
\end{proof}
\begin{prop}
\label{pro:positivity2}Let $\varphi\in\mathcal{Z}$, then\begin{equation}
0\leq\Tr[\tilde{\rho}(a^{*}(\varphi)+a(\varphi))^{2}]=2\rRe(\alpha^{*}(\varphi^{\vee2}))+\Tr[(2\gamma+\mathbf{1})\varphi\varphi^{*}]\label{eq:positivity2}\end{equation}
and $|2\rRe(\alpha^{*}(\varphi^{\cdot\vee2}))|\leq\Tr[(2\gamma+\mathbf{1})\varphi\varphi^{*}]$.
\end{prop}
This condition is used with the three components of $\vec{\varphi}=\vec{G}+\vec{k}f$.
\begin{proof}
A computation using the canonical commutation relations yields
\begin{align*}
Tr & [\tilde{\rho}\,(a^{*}(\varphi)+a(\varphi))^{2}]\\
 & =\Tr[\tilde{\rho}\,(a^{*}(\varphi))^{2}+\tilde{\rho}\,(a(\varphi))^{2}+\tilde{\rho}\,(a^{*}(\varphi)a(\varphi)+a(\varphi)a^{*}(\varphi))]\\
 & =\alpha^{*}\varphi^{\vee2}+\varphi^{\vee2*}\alpha+\Tr[\gamma\,\varphi\varphi^{*}+(\gamma+\mathbf{1})\psi\psi^{*}].
\end{align*}
\end{proof}

\section{\label{sec:Minimization-coherent-states}Minimization over Coherent States}

For this section we can take $\sigma=0$ if we consider the parameter
$f$ in the energy to be in~$\tcZ := L^{2}(S_{\sigma,\Lambda}\times\mathbb{Z}_{2},(\tfrac{1}{2}|\vec{k}|^{2}+|\vec{k}|)dk)$.
Recall that $S_{\sigma,\Lambda} = \{\vec{k}\in\mathbb{R}^{3}\,|\,\sigma\leq|\vec{k}|\leq\Lambda\}$.
\begin{rem}
For a coherent state (see Definition~\ref{def:coherent-states})
the energy reduces to
\begin{equation}
\mathcal{E}_{g,\vec{p}}(f)=\frac{1}{2}\|\vec{G}\|^{2}+\frac{1}{2}(f^{*}\vec{k}f+2\rRe(f^{*}\vec{G})-\vec{p})^{\cdot2}+f^{*}(\tfrac{1}{2}|\vec{k}|^{2}+|\vec{k}|)f
\,.\label{eq:energy_coherent-2}
\end{equation}
Note that, for $\sigma>0$, $\cZ = L^{2}(S_{\sigma,\Lambda}\times\mathbb{Z}_{2},dk) = \tcZ$, while for 
$\sigma=0$, $\cZ \subset \tcZ$, and $\mathcal{E}_{g,\vec{p}}(f)$ extends to~$\tcZ$ by using Equation~(\ref{eq:energy_coherent-2}).
\end{rem}
\begin{thm}
\label{thm:fixed-point-f-u}There exists a universal constant~$C<\infty$
such that, for $0\leq\sigma<\Lambda<\infty$, $g^{2}\ln(\Lambda+2)\leq C$
and $|\vec{p}|\leq1/3$, there exists a unique~$f_{\vec{p}}$ which
minimizes~$\mathcal{E}_{g,\vec{p}}$ in~$\tcZ$.
\begin{enumerate}
\item \label{enu:closed-system-f-u-1}The minimizer $f_{\vec{p}}$ solves
the system of equations\begin{align}
f_{\vec{p}} & =\frac{\vec{u}_{\vec{p}}\cdot\vec{G}}{\frac{1}{2}|\vec{k}|^{2}+|\vec{k}|-\vec{k}\cdot\vec{u}_{\vec{p}}}\,,\label{eq:f_u_bis1-1}\\
\vec{u}_{\vec{p}} & =\vec{p}-2\rRe(f_{\vec{p}}^{*}\vec{G})-f_{\vec{p}}^{*}\vec{k}f_{\vec{p}}\,,\label{eq:f_u_bis2-1}\end{align}
with $|\vec{u}_{\vec{p}}|\leq|\vec{p}|$.
\item For~$0\leq\sigma<\Lambda<\infty$,\[
\inf_{f\in \cZ} \mathcal{E}_{g,\vec{p}}(f) 
\ = \ 
\inf_{f\in \tcZ} \mathcal{E}_{g,\vec{p}}(f)
\ = \ 
\mathcal{E}_{g,\vec{p}}(f_{\vec{p}})
\,,\]
and for~$0<\sigma<\Lambda<\infty$, we have that $f_{\vec{p}} \in \cZ$.
\item \label{enu:asymptotic-devel-energy-1} For fixed~$g$, $\sigma$,
$\Lambda$, as a function of $\vec{p}$, \[
\mathcal{E}_{g,\vec{p}}(f_{\vec{p}})=\mathcal{E}_{g,\vec{p}}(0)-\vec{p}\cdot\vec{G}^{*}\frac{1}{\frac{1}{2}|\vec{k}|^{2}+|\vec{k}|+2\vec{G}\cdot\vec{G}^{*}}\vec{G}\cdot\vec{p}+\mathcal{O}(|\vec{p}|^{3})\,.\]

\item \label{enu:devel-energy-in-fp-1}For all~$f$ in~$\tcZ$,\begin{multline}
\mathcal{E}_{g,\vec{p}}(f_{\vec{p}}+f)=\mathcal{E}_{g,\vec{p}}(f_{\vec{p}})+f^{*}(\tfrac{1}{2}|\vec{k}|^{2}+|\vec{k}|-\vec{u}_{\vec{p}}\cdot\vec{k})f\\
+\tfrac{1}{2}\big(f^{*}\vec{k}f+2\rRe(f_{\vec{p}}^{*}\vec{k}f)+2\rRe(f^{*}\vec{G})\big)^{\cdot 2}\,.\label{eq:devel-energy-in-fp-1}\end{multline}

\item \label{enu:comparision-energies-0-fp-1}The energy $\mathcal{E}_{g,\vec{p}}(f_{\vec{p}})$ of the minimizer
compared to the energy of the vacuum state~$\mathcal{E}_{g,\vec{p}}(0)$ is
\[
\mathcal{E}_{g,\vec{p}}(f_{\vec{p}})
\ = \ 
\mathcal{E}_{g,\vec{p}}(0)-\tfrac{1}{2}2\rRe(f_{\vec{p}}^{*}\vec{u}_{\vec{p}}\cdot\vec{G})-\tfrac{1}{2}|\vec{u}_{\vec{p}}-\vec{p}|^{2}\,.\]
Note that the term $2\rRe(f_{\vec{p}}^{*}\vec{u}_{\vec{p}}\cdot\vec{G})$ is non-negative.
\end{enumerate}
\end{thm}

\begin{rem}
Our hypotheses are similar those of Chen, Fröhlich, and Pizzo~\cite{MR2585987}, where their vector $\vec{\nabla}E_{\vec{p}}^{\sigma}$
is analogous to~$\vec{u}_{\vec{p}}$ in our notations.

The construction of $\vec{u}_{\vec{p}}$ as the solution of a fixed
point problem and the dependency in the parameter $\vec{p}$ imply
that the map $\vec{p}\mapsto\vec{u}_{\vec{p}}$ is of class $\mathcal{C}^{\infty}$.
\end{rem}

\begin{rem}
We note that we also expect to have $\vec{u}_{\vec{p}}$ in the neighboorhood
of~$\vec{p}$.
\end{rem}

\begin{rem}
The minimizer is constructed as the solution of a fixed point problem.
As a result the application \[
(\sigma,\Lambda,g,\vec{p})\mapsto\inf_{\rho\in\mbox{coh}}\Tr[H_{g,\vec{p}}\,\rho]\]
is continuous on the domain defined by Theorem~\ref{thm:fixed-point-f-u},
and at $\sigma$, $\Lambda$ fixed,\[
(g,\vec{p})\mapsto\inf_{\rho\in\mbox{coh}}\Tr[H_{g,\vec{p}}\,\rho]\]
is analytic for $g^{2}<C/\ln(\Lambda+2)$ and $|\vec{p}|<1/3$.\end{rem}
\begin{proof}[Proof of Theorem~\ref{thm:fixed-point-f-u}]
\emph{Proof of~\ref{enu:closed-system-f-u-1}.} Assume there is
a point~$f_{\vec{p}}$ where the minimum is attained. The partial
derivative of the energy at the point~$f_{\vec{p}}$\begin{multline*}
\partial_{f^{*}}\mathcal{E}(f_{\vec{p}})\\
=((f_{\vec{p}}^{*}\vec{k}f_{\vec{p}}-\vec{p}+2\rRe(f_{\vec{p}}^{*}\vec{G}))\cdot\vec{k}+\frac{1}{2}|\vec{k}|^{2}+|\vec{k}|)f_{\vec{p}}-(\vec{p}-f_{\vec{p}}^{*}\vec{k}f_{\vec{p}}-2\rRe(f_{\vec{p}}^{*}\vec{G}))\cdot\vec{G}\end{multline*}
then vanishes, where the derivative $\partial_{f^{*}}\mathcal{E}(f)$
at a point~$f$ is the unique vector in $\tcZ^*\cong L^2(S_{\sigma,\Lambda},(\tfrac{1}{2}|\vec k|^2+|\vec k|)^{-1}dk)$
defined by \[
\mathcal{E}(f+\delta f)-\mathcal{E}(f)=2\rRe(\delta\! f^{*}\,\partial_{f^{*}}\mathcal{E}(f))+o(\|\delta f\|_{\tcZ})\]
with $f,\,\delta f \in \tcZ$.
Observe that\begin{align*}
0 & \leq\mathcal{E}_{g,\vec{p}}(0)-\mathcal{E}_{g,\vec{p}}(f_{\vec{p}})\\
 & =\frac{1}{2}|\vec{p}|^{\cdot2}-\frac{1}{2}(f_{\vec{p}}^{*}\vec{k}f_{\vec{p}}+2\rRe(f_{\vec{p}}^{*}\vec{G})-\vec{p})^{\cdot2}-f_{\vec{p}}^{*}(\frac{1}{2}|\vec{k}|^{2}+|\vec{k}|)f_{\vec{p}}\end{align*}
and hence $|\vec{p}|\geq|\vec{u}_{\vec{p}}|$ with $\vec{u}_{\vec{p}}:=\vec{p}-f_{\vec{p}}^{*}\vec{k}f_{\vec{p}}-2\rRe(f_{\vec{p}}^{*}\vec{G})$.
Since $|\vec{u}_{\vec{p}}|\leq|\vec{p}|<1$, it makes sense to write\[
f_{\vec{p}}=\frac{\vec{u}_{\vec{p}}\cdot\vec{G}}{\frac{1}{2}|\vec{k}|^{2}+|\vec{k}|-\vec{u}_{\vec{p}}\cdot\vec{k}}\,.\]
Hence the minimum point~$f_{\vec{p}}$ satisfies Equations~(\ref{eq:f_u_bis1-1})
and~(\ref{eq:f_u_bis2-1}). It is in particular sufficient to prove
that there exist a unique~$\vec{u}_{\vec{p}}$ in a ball~$\bar{B}(0,r)$
with~$r\geq|\vec{p}|$ such that the function in Equation~(\ref{eq:f_u_bis1-1})
satisfies also Equation~(\ref{eq:f_u_bis2-1}) to prove the existence
and uniqueness of a minimizer.

\emph{Proof of the existence and uniqueness of a solution.} Let $\tfrac{1}{3}<r<1$,
$\vec{u}\in\mathbb{R}^{3}$, $|\vec{u}|\leq r<1$ and\[
\Phi_{\vec{u}}(\vec{k})=\frac{\vec{u}\cdot\vec{G}(\vec{k})}{\frac{1}{2}|\vec{k}|^{2}+|\vec{k}|-\vec{k}\cdot\vec{u}}\,.\]
Observe that $\Phi_{\vec{u}}\in \tilde{\mathcal Z}$
, indeed, if $|\vec{u}|<1$ then $\frac{1}{2}|\vec{k}|^{2}+|\vec{k}|-\vec{k}\cdot\vec{u}\geq(1-r)(\frac{1}{2}|\vec{k}|^{2}+|\vec{k}|)$,
and with $\vec{\varepsilon}(\vec{k})=\vec{\varepsilon}(\vec{k},+)+\vec{\varepsilon}(\vec{k},-)$,\begin{align*}
\int_{|\vec{k}|\in[\sigma,\Lambda]}(\tfrac{1}{2}|\vec{k}|^{2}+|\vec{k}|)|\Phi_{\vec{u}}(\vec{k})|^{2}dk & \leq g^{2}\int_{|\vec{k}|\in[\sigma,\Lambda]}\frac{1}{|\vec{k}|}\frac{1}{(1-r)^{2}}\frac{|\vec{u}\cdot\vec{\varepsilon}(\vec{k})|^{2}}{\tfrac{1}{2}|\vec{k}|^{2}+|\vec{k}|}dk<+\infty\,.\\
 & \leq C_{0}g^{2}\ln(\Lambda+2)\frac{|\vec{u}|^{2}}{(1-r)^{2}}\end{align*}
for some universal constant $C_{0}>0$. Observe then that\[
\int_{|\vec{k}|\in[\sigma,\Lambda]}\frac{|\vec{G}(k)|^{2}}{\tfrac{1}{2}|\vec{k}|^{2}+|\vec{k}|}dk\leq C_{0}g^{2}\ln(\Lambda+2)\]
for some universal constant $C_{0}>0$. It follows that $\Phi_{\vec{u}}^{*}\vec{G}\in L^{1}(S_{\sigma,\Lambda}\times\mathbb{Z}_{2})$.
Note that if $\sigma=0$ then $\Phi_{\vec{u}}\notin L^{2}(S_{\sigma,\Lambda}\times\mathbb{Z}_{2})$
(for $\vec{u}\neq0$).

We can thus define the application\begin{align*}
\bar{B}(0,r)\ni\vec{u} & \mapsto\vec{\Psi}(\vec{u}):=\vec{p}-\Phi_{\vec{u}}^{*}\vec{k}\Phi_{\vec{u}}-2\rRe(\Phi_{\vec{u}}^{*}\vec{G})\in\mathbb{R}^{3}\,.\end{align*}
We check that the hypotheses of the Banach-Picard fixed point theorem
are verified on the ball $\bar{B}(0,r)$, which will prove the result.

\emph{Stability}: If $g^{2}\ln(\Lambda+2)$ is sufficiently small,
we get from\[
|\vec{\Psi}(\vec{u})|\leq|\Phi_{\vec{u}}^{*}\vec{k}\Phi_{\vec{u}}|+|2\rRe(\Phi_{\vec{u}}^{*}\vec{G})|+|\vec{p}|\]
and the estimates above that the sum of the two first terms is smaller
than $r-1/3$ and since $|\vec{p}|\leq1/3$ the map~$\vec{\Psi}$
sends $\bar{B}(0,r)$ into itself,\[
\vec{\Psi}(\bar{B}(0,r))\subseteq\bar{B}(0,r)\,.\]

\emph{Contraction}: For $\vec{u}$ and $\vec{v}$ in $\bar{B}(0,r)$,
we have that \begin{align*}
|\Phi_{\vec{u}} & (\vec{k})-\Phi_{\vec{v}}(\vec{k})|(\tfrac{1}{2}|\vec{k}|^{2}+|\vec{k}|)\\
 & =|\frac{\vec{u}.\vec{G}(\vec{k})}{\frac{1}{2}|\vec{k}|^{2}+|\vec{k}|-\vec{k}\cdot\vec{u}}-\frac{\vec{v}.\vec{G}(\vec{k})}{\frac{1}{2}|\vec{k}|^{2}+|\vec{k}|-\vec{k}.\vec{v}}|(\tfrac{1}{2}|\vec{k}|^{2}+|\vec{k}|)\\
 & \leq\Big(\frac{|\vec{u}-\vec{v}||\vec{G}(\vec{k})|}{\frac{1}{2}|\vec{k}|^{2}+|\vec{k}|-\vec{k}\cdot\vec{u}}\\
 &\quad+|\vec{v}||\vec{G}(\vec{k})|\,|\frac{1}{\frac{1}{2}|\vec{k}|^{2}+|\vec{k}|-\vec{k}\cdot\vec{v}}-\frac{1}{\frac{1}{2}|\vec{k}|^{2}+|\vec{k}|-\vec{k}\cdot\vec{u}}|\Big)(\tfrac{1}{2}|\vec{k}|^{2}+|\vec{k}|)\\
 & \leq|\vec{u}-\vec{v}||\vec{G}(\vec{k})|\frac{1}{(1-r)}(1+\frac{r|\vec{k}|}{\frac{1}{2}|\vec{k}|^{2}+|\vec{k}|(1-r)})\\
 & \leq|\vec{u}-\vec{v}||\vec{G}(\vec{k})|\frac{1}{(1-r)^{2}}\,.\end{align*}
For the term $2\rRe(\Phi_{\vec{u}}^{*}\vec{G})$, we observe that\begin{align*}
|2\rRe & (\Phi_{\vec{u}}^{*}\vec{G})-2\rRe(\Phi_{\vec{v}}^{*}\vec{G})|\\
 & \leq g^{2}2|\vec{u}-\vec{v}|\frac{1}{(1-r)^{2}}\int_{|\vec{k}|\in[\sigma,\Lambda]}\frac{1}{|\vec{k}|}\frac{1}{\frac{1}{2}|\vec{k}|^{2}+|\vec{k}|}d^{3}k\\
 & \leq C_{1}g^{2}\ln(2+\Lambda)2|\vec{u}-\vec{v}|\frac{1}{(1-r)^{2}}\,.\end{align*}
Note that, for $g^{2}\ln(2+\Lambda)<(1-r)^{2}/(3C_{1})$,\[
|2\rRe(\Phi_{\vec{u}}^{*}\vec{G})-2\rRe(\Phi_{\vec{v}}^{*}\vec{G})|<\frac{1}{3}|\vec{u}-\vec{v}|\,.\]
Finally, for the term $\Phi_{\vec{u}}^{*}\vec{k}\Phi_{\vec{u}}$,
we obtain the estimate\begin{align*}
|\Phi_{\vec{u}}^{*} & \vec{k}\Phi_{\vec{u}}-\Phi_{\vec{v}}^{*}\vec{k}\Phi_{\vec{v}}|\\
 & \leq\int_{|\vec{k}|\in[\sigma,\Lambda]}(\tfrac{1}{2}|\vec{k}|^{2}+|\vec{k}|)|\Phi_{\vec{u}}(\vec{k})-\Phi_{\vec{v}}(\vec{k})|(|\Phi_{\vec{u}}(\vec{k})|+|\Phi_{\vec{v}}(\vec{k})|)d^{3}k\\
 & \leq\frac{|\vec{u}-\vec{v}|}{(1-r)^{2}}\int_{|\vec{k}|\in[\sigma,\Lambda]}|\vec{G}(\vec{k})|(|\Phi_{\vec{u}}(\vec{k})|+|\Phi_{\vec{v}}(\vec{k})|)d\vec{k}\\
 & \leq\frac{|\vec{u}-\vec{v}|}{(1-r)^{2}}\|(\tfrac{1}{2}|\vec{k}|^{2}+|\vec{k}|)^{-1/2}G\|(\|\sqrt{\tfrac{1}{2}|\vec{k}|^{2}+|\vec{k}|}\Phi_{\vec{u}}\|+\|\sqrt{\tfrac{1}{2}|\vec{k}|^{2}+|\vec{k}|}\Phi_{\vec{v}}\|)\\
 & \leq C_{2}|\vec{u}-\vec{v}|\,(|\vec{u}|+|\vec{v}|)g^{2}\ln(\Lambda+2)\,,\end{align*}
and thus this term can be controlled for $|g\ln(\Lambda+2)|^{2}$
sufficiently small by $\frac{1}{3}|\vec{u}-\vec{v}|$. We thus get
a contraction\[
|\vec{\Psi}(\vec{u})-\vec{\Psi}(\vec{u}^{\prime})|\leq\frac{2}{3}|\vec{u}-\vec{u}^{\prime}|\]
 and with $f_{\vec{p}}=\Phi_{\vec{u}_{\vec{p}}}$ Equation~(\ref{eq:f_u_bis1-1})
is solved.

\emph{Proof of \ref{enu:asymptotic-devel-energy-1}.} The expression
of the energy $\mathcal{E}_{g,\vec{p}}(f)$ given in Equation~(\ref{eq:energy_coherent-2})
implies that $\mathcal{E}_{g,\vec{p}}(f)\geq\frac{1}{2}\|\vec{G}\|^{2}$,
and for $\vec{p}=\vec{0}$ this minimum is only attained at the point
$f_{\vec{0}}=0$. It follows that $f_{\vec{p}}=\partial_{\vec{p}}f{}_{\vec{0}}\cdot\vec{p}+\mathcal{O}(|\vec{p}|^{2})\,.$
From Equation~(\ref{eq:f_u_bis2-1}) we deduce\[
\vec{u}_{\vec{p}}=\vec{p}-2\rRe((\partial_{\vec{p}}f{}_{\vec{0}}\cdot\vec{p})^{*}\vec{G})+\mathcal{O}(|\vec{p}|^{2})\]
and thus\begin{align*}
f_{\vec{p}} & =\frac{(\vec{p}-2\rRe((\partial_{\vec{p}}f{}_{\vec{0}}\cdot\vec{p})^{*}\vec{G})).\vec{G}}{\frac{1}{2}|\vec{k}|^{2}+|\vec{k}|-\vec{k}\cdot\vec{u}_{\vec{p}}}+\mathcal{O}(|\vec{p}|^{2})\\
 & =(\frac{1}{2}|\vec{k}|^{2}+|\vec{k}|)^{-1}(\vec{p}-2\rRe((\partial_{\vec{p}}f{}_{\vec{0}}\cdot\vec{p})^{*}\vec{G}))\cdot\vec{G}+\mathcal{O}(|\vec{p}|^{2})\,.\end{align*}
Expanding the left hand side of this equality in $\vec{0}$ brings\[
\partial_{\vec{p}}f{}_{\vec{0}}\cdot\vec{p}=(\frac{1}{2}|\vec{k}|^{2}+|\vec{k}|)^{-1}(\vec{p}-2\rRe((\partial_{\vec{p}}f{}_{\vec{0}}\cdot\vec{p})^{*}\vec{G}))\cdot\vec{G}\]
and hence $\partial_{\vec{p}}f{}_{\vec{0}}=(\frac{1}{2}|\vec{k}|^{2}+|\vec{k}|+2\vec{G}\cdot\vec{G}^{*})^{-1}\vec{G}$.
The expansion of $f_{\vec{p}}$ to the second order is then\[
f_{\vec{p}}=(\frac{1}{2}|\vec{k}|^{2}+|\vec{k}|+2\vec{G}\cdot\vec{G}^{*})^{-1}\vec{G}\cdot\vec{p}+\mathcal{O}(|\vec{p}|^{2})\,.\]
We can compute the energy modulo error terms in~$\mathcal{O}(|\vec{p}|^{3})$.
To have less heavy computations we set $A=\frac{1}{2}|\vec{k}|^{2}+|\vec{k}|+2\vec{G}\cdot\vec{G}^{*}$
and get \begin{align*}
\mathcal{E}_{g,\vec{p}} & (f_{\vec{p}})-\frac{1}{2}\|\vec{G}\|^{2}-\frac{1}{2}|\vec{p}|^{2}\\
 & \equiv-\frac{1}{2}|\vec{p}|^{2}+\frac{1}{2}(2\rRe(\vec{p}\cdot\partial_{\vec{p}}f_{\vec{0}}^{*}\vec{G})-\vec{p})^{\cdot2}+\vec{p}\cdot\partial_{\vec{p}}f_{\vec{0}}^{*}(\frac{1}{2}|\vec{k}|^{2}+|\vec{k}|)\partial_{\vec{p}}f_{\vec{0}}^{*}\cdot\vec{p}\\
 & \equiv\frac{1}{2}(2\rRe(\vec{p}\cdot\vec{G}^{*}A^{-1}\vec{G}))^{\cdot2}-2\vec{p}\cdot\vec{G}^{*}A^{-1}\vec{G}\cdot\vec{p}\\
 & \phantom{=}\qquad+\vec{p}\cdot\vec{G}^{*}A^{-1}(\frac{1}{2}|\vec{k}|^{2}+|\vec{k}|)A^{-1}\vec{G}\cdot\vec{p}\\ \allowdisplaybreaks
 & \equiv2(\vec{p}\cdot\vec{G}^{*}A^{-1}\vec{G})^{\cdot2}+\vec{p}\cdot\vec{G}^{*}A^{-1}((\frac{1}{2}|\vec{k}|^{2}+|\vec{k}|)-2A)A^{-1}\vec{G}\cdot\vec{p}\\
 & \equiv\vec{p}\cdot\vec{G}^{*}A^{-1}2\vec{G}\cdot\vec{G}^{*}A^{-1}\vec{G}\cdot\vec{p}-\vec{p}\cdot\vec{G}^{*}A^{-1}(\frac{1}{2}|\vec{k}|^{2}+|\vec{k}|+4\vec{G}\cdot\vec{G}^{*}))A^{-1}\vec{G}\cdot\vec{p}\\ 
 & \equiv-\vec{p}\cdot\vec{G}^{*}(\frac{1}{2}|\vec{k}|^{2}+|\vec{k}|+2\vec{G}\cdot\vec{G}^{*})^{-1}\vec{G}\cdot\vec{p}\end{align*}
which yields the result.

\emph{Proof of~\ref{enu:devel-energy-in-fp-1}.} The Taylor expansion
of the energy around $f_{\vec{p}}$ is\begin{align*}
\mathcal{E}_{g,\vec{p}}(f_{\vec{p}}+f) & =\mathcal{E}_{g,\vec{p}}(f_{\vec{p}})+f^{*}\,\partial_{f^{*}}\mathcal{E}(f_{\vec{p}})+\partial_{f}\mathcal{E}(f_{\vec{p}})\, f\\
 & \phantom{=}+\frac{1}{2}\Big\{(f^{*}\vec{k}f+2\rRe(f^{*}\vec{G})+2\rRe(f_{\vec{p}}^{*}\vec{k}f))^{.2}\\
 & \phantom{=+\frac{1}{2}\Big\{}+2(f_{\vec{p}}^{*}\vec{k}f_{\vec{p}}+2\rRe(f_{\vec{p}}^{*}\vec{G})-\vec{p})\cdot f^{*}\vec{k}f+f^{*}|\vec{k}|^{2}f\Big\}+f^{*}|\vec{k}|f\,.\end{align*}
Since $\partial_{f^{*}}\mathcal{E}(f_{\vec{p}})$ vanishes this gives
Equation~(\ref{eq:devel-energy-in-fp-1}).

\emph{Proof of \ref{enu:comparision-energies-0-fp-1}.} It is sufficient
to replace $f$ by $-f_{\vec{p}}$ in Equation~(\ref{eq:devel-energy-in-fp-1}).
The observation\[
f_{\vec{p}}^{*}\vec{u}_{\vec{p}}\cdot\vec{G}=\int\frac{(\vec{u}_{\vec{p}}\cdot\vec{G}(\vec{k}))^{2}dk}{\frac{1}{2}|\vec{k}|^{2}+|\vec{k}|-\vec{k}\cdot\vec{u}_{\vec{p}}}\]
shows that $2\rRe(f_{\vec{p}}^{*}\vec{u}_{\vec{p}}\cdot\vec{G})$
is non-negative since $|\vec{u}_{\vec{p}}|<1$.
\end{proof}

\section{\label{sec:Existence-and-uniqueness}The Minimizer for the Energy Functional varying over Pure Quasifree
States}
\begin{defn}
Let $\mathcal{Z}$ be a $\mathbb{C}$-Hilbert space. Let $Y$ be the
$\mathbb{R}$-Hilbert space of antilinear operators $\hat{r}$ on
$\mathcal{Z}$, self-adjoint in the sense that $\forall z,z^{\prime}\in\mathcal{Z},\,\langle z,\hat{r}z^{\prime}\rangle=\langle z^{\prime},\hat{r}z\rangle$,
and Hilbert-Schmidt in the sense that the positive operator~$\hat{r}^{2}$
is trace class. The space $X=\mathcal{Z}\times Y$ with the scalar
product\[
\langle(f,\hat{r}),(f',\hat{r}')\rangle_{X}=f^{*}f'+\Tr[\hat{r}\hat{r}']\]
is an $\mathbb{R}$-Hilbert space.

Keeping $\sigma >0$, we only need to use
$\mathcal{Z}=L^{2}(S_{\sigma,\Lambda}\times\mathbb{Z}_{2})$ in this
section.\end{defn}
\begin{thm}
\label{pro:existence-uniqueness}Let $0<\sigma<\Lambda<\infty$. There
exists $C>0$ such that for $g,|\vec{p}|\leq C$ there exists a unique
minimizer for $\hat{\mathcal{E}}_{g,\vec{p}}(f,\hat{r})$.\end{thm}
\begin{proof}
This result follows from convexity and coercivity arguments. By Proposition~\ref{pro:convexity_E(f,r)}, $\hat{\mathcal{E}}_{g,\vec{p}}(f,\hat{r})$
is strictly $\theta$-convex (i.e., uniformly strictly convex) on $\bar{B}_{X}(0,R)$ for some $R>0$ and $\theta>0$.
Since $\hat{\mathcal{E}}_{g,\vec{p}}(f,\hat{r})$
is strongly continous on the closed and convex set $\bar{B}_{X}(0,R)$ of the Hilbert space $X$ we get the existence and uniqueness of a minimizer in $\bar{B}_{X}(0,R)$. (See for example \cite{MR2326223}. The uniform strict convexity allows to prove directly that a minimizing sequence is a Cauchy sequence.) Proposition~\ref{pro:coercivity_E(f,r)} then proves that it is the only minimum of $\hat{\mathcal{E}}_{g,\vec{p}}(f,\hat{r})$ on the whole space.

%
%

Note that to use Propositions~\ref{pro:convexity_E(f,r)} and~\ref{pro:coercivity_E(f,r)}
we need to restrict to values of $g$ and $|\vec{p}|$ smaller than
some constant~$C>0$.\end{proof}
\begin{prop}[Convexity]
\label{pro:convexity_E(f,r)}There exist $0<C,R<\infty$ such that
for $g\leq C$ and $|\vec{p}|\leq\frac{1}{2}$, the Hessian of the
energy satisfies $\mathcal{H}\hat{\mathcal{E}}_{g,\vec{p}}(f,\hat{r})\geq\frac{\sigma}{4}\mathbf{1}_{X}$
on the ball $B_{X}(0,R)$.

\end{prop}
\begin{proof}
We use that strict positivity of the Hessian implies strict convexity
and thus first compute the Hessian in $(0,0)$. The Hessian $\mathcal{H}\hat{\mathcal{E}}_{g,\vec{p}}(f,\hat{r})\in\mathcal{B}(X)$
is defined using the Fr\'echet derivative\begin{multline*}
\hat{\mathcal{E}}_{g,\vec{p}}(f+\delta f,\hat{r}+\delta\hat{r})-\hat{\mathcal{E}}_{g,\vec{p}}(f,\hat{r})\\
\ = \ D\hat{\mathcal{E}}_{g,\vec{p}}(f,\hat{r})(\delta f,\delta\hat{r}) +
\frac{1}{2}\big\langle(\delta f,\delta\hat{r})\,,\,\mathcal{H}\hat{\mathcal{E}}_{g,\vec{p}}(f,\hat{r})\,(\delta f,\hat{r})\big\rangle_{X}
+o(\|(\delta f,\delta\hat{r})\|_{X}^{2})
\end{multline*}
with $D\hat{\mathcal{E}}_{g,\vec{p}}(0,0)\in\mathcal{B}(X,\mathbb{R})$.
(Note that differentiability is granted in this case because $|\vec{k}|\leq\Lambda<\infty$.)
For any $\mu>0$, $\forall(f,\hat{r})\in X$,\begin{align*}
\langle(f, & \hat{r})\,,\,\frac{1}{2}\mathcal{H}\hat{\mathcal{E}}_{g,\vec{p}}(0,0)\,(f,\hat{r})\rangle_{X}\\
 & =2\rRe\langle\hat{r}\vec{k}f;\vec{G}\rangle+\frac{1}{2}(2\rRe(f^{*}\vec{G}))^{\cdot2}+\Tr[\hat{r}^{2}\vec{G}\cdot\vec{G}^{*}]\\
 & \phantom{=}+\frac{1}{2}\big\{\Tr[\hat{r}\vec{k}\cdot\hat{r}\vec{k}]+\Tr[|\vec{k}|^{2}\hat{r}^{2}]\big\}\\
 & \phantom{=}+\Tr[\hat{r}^{2}(|\vec{k}|-\vec{k}\cdot\vec{p})]+f^{*}(\frac{1}{2}|\vec{k}|^{2}+|\vec{k}|-\vec{k}\cdot\vec{p})f\allowdisplaybreaks\\
 & \geq\Tr[\hat{r}^{2}\vec{G}\cdot\vec{G}^{*}]-\mu\|\hat{r}\vec{G}\|^{2}-\frac{1}{\mu}\|\vec{k}f\|^{2}\\
 & \phantom{\geq}+\frac{1}{2}\big\{(2\rRe(\delta f^{*}\vec{G}))^{\cdot2}+\Tr[\hat{r}\vec{k}\cdot\hat{r}\vec{k}]+\Tr[|\vec{k}|^{2}\hat{r}^{2}]\big\}\\
 & \phantom{\geq}+\Tr[\hat{r}^{2}(|\vec{k}|-\vec{k}\cdot\vec{p})]+f^{*}(\frac{1}{2}|\vec{k}|^{2}+|\vec{k}|-\vec{k}\cdot\vec{p})f\\
 & \geq\Tr[\hat{r}^{2}(|\vec{k}|-\vec{k}\cdot\vec{p}+(1-\mu)\vec{G}\cdot\vec{G}^{*})]+f^{*}((\frac{1}{2}-\frac{1}{\mu})|\vec{k}|^{2}+|\vec{k}|-\vec{k}\cdot\vec{p})f\,,\end{align*}
since\[
|2\rRe\langle\hat{r}\vec{k}f;\vec{G}\rangle|\leq2\|\hat{r}\vec{G}\|\|\vec{k}f\|=2\sqrt{\mu}\|\hat{r}\vec{G}\|\frac{1}{\sqrt{\mu}}\|\vec{k}f\|\leq\mu\|\hat{r}\vec{G}\|^{2}+\frac{1}{\mu}\|\vec{k}f\|^{2}\,.\]
With $\mu=2$ we obtain (with $|\vec{p}|\leq\frac{1}{2}$)\begin{align*}
\frac{1}{2}\mathcal{H}\hat{\mathcal{E}}_{g,\vec{p}}(0,0)(f,\hat{r}) & \geq\Tr[\hat{r}^{2}(|\vec{k}|-\vec{k}\cdot\vec{p}-\vec{G}\cdot\vec{G}^{*})]+f^{*}(|\vec{k}|-\vec{k}\cdot\vec{p})f\\
 & \geq\Tr[\hat{r}^{2}(|\vec{k}|(1-\||\vec{k}|^{-1/2}\vec{G}\|^{2})-\vec{k}\cdot\vec{p})]+f^{*}(|\vec{k}|-\vec{k}\cdot\vec{p})f\\
 & \geq\Tr[\hat{r}^{2}\sigma(\frac{1}{2}-\||\vec{k}|^{-1/2}\vec{G}\|^{2})]+f^{*}\frac{\sigma}{2}f\end{align*}
and for $g$ small enough\[
\frac{1}{2}\mathcal{H}\hat{\mathcal{E}}_{g,\vec{p}}(0,0)\geq\frac{\sigma}{4}\,.\]
We then compare it with the Hessian in points near zero. Observing
that the Hessian is continous with respect to $(f,\hat{r},\vec{p},g)$,
we deduce that there exist~$R< \infty$ and~$C > 0$, as asserted.
\end{proof}

\begin{prop}[Coercivity]
\label{pro:coercivity_E(f,r)}Suppose $\vec{p}$ and $C>0$ are fixed
such that $\frac{1}{2}|\vec{p}|^{2}+\frac{1}{2}\|\vec{G}\|^{2}<\sigma R^{2}$,
with the value of $R$ given by Proposition~\ref{pro:convexity_E(f,r)},
for any $0<g<C$. For every $(f,\hat{r})\in X$,
\[
\hat{\mathcal{E}}_{g,\vec{p}}(f,\hat{r}) 
\ \geq \ 
\Tr[\hat{r}^{2}|\vec{k}|]+f^{*}|\vec{k}|f
\ \geq \
\sigma \, \big\| (f,\hat{r}) \big\|_X^2 \, . 
\]
Since $\hat{\mathcal{E}}_{g,\vec{p}}(0,0)=\frac{1}{2}|\vec{p}|^{2}+\frac{1}{2}\|\vec{G}\|^{2}<\sigma R^{2}$,
any minimizing sequence takes its values in $\bar{B}_{X}(0,R)$. 
\end{prop}

\section{\label{sec:Perturbative-approach}Asymptotics for small Coupling
and Momentum}

We use below an identification between self-adjoint $\mathbb{C}$-antilinear
Hilbert-Schmidt operator~$\hat{r}$ and symmetric two vector $r$
given by the relation $\langle\varphi,\hat{r}\psi\rangle_{\mathcal{Z}}=\langle\varphi\otimes\psi,r\rangle_{\mathcal{Z}^{\otimes2}}$.
Note that the self-adjointness condition for~$\hat{r}$ is equivalent
to the symmetry condition~$r\in\mathcal{Z}^{\vee2}$.
\begin{thm}
\label{thm:perturbative}Let~$0<\sigma<\Lambda<\infty$. There exists
$C>0$ such that for $|g|,|\vec{p}|<C$, there exist two functions
$f_{g,\vec{p}}$ and $\hat{r}_{g,\vec{p}}$ which are smooth in $(g,\vec{p})$
such that the minimum of the energy $\hat{\mathcal{E}}_{g,\vec{p}}(f,\hat{r})$ is attained at $(f_{g,\vec{p}},\hat{r}_{g,\vec{p}})$.
These functions satisfy\begin{align*}
f_{g,\vec{p}} & =\frac{\vec{p}.\vec{G}}{\frac{1}{2}|\vec{k}|^{2}+|\vec{k}|}+\mathcal{O}(\|(g,\vec{p})\|^{3})\\
r_{g,\vec{p}} & =-S^{-1}\vec{G}\cdot\vee\vec{G}+\mathcal{O}(\|(g,\vec{p})\|^{3})\,,\end{align*}
\textup{with $S=\vec{k}\cdot\otimes\vec{k}+2(\frac{1}{2}|\vec{k}|^{2}+|\vec{k}|)\vee\mathbf{1}_{\mathcal{Z}}$.
As a consequence}\begin{multline*}
E_{BHF}(g,\vec{p},\sigma,\Lambda)\\
=\hat{\mathcal{E}}_{g,\vec{p}}(0_X)-\vec{p}\cdot\vec{G}^{*}\frac{1}{\frac{1}{2}|\vec{k}|^{2}+|\vec{k}|}\vec{G}\cdot\vec{p}-\frac{1}{2}\vec{G}^{\cdot\vee2*}S^{-1}\vec{G}^{\cdot\vee2}+\mathcal{O}(\|(g,\vec{p})\|^{5})\,.\end{multline*}
\end{thm}
\begin{rem}
The energy in $0_X$ is the energy of the vacuum state and
is $\hat{\mathcal{E}}_{g,\vec{p}}(0_X) = \frac{1}{2}\vec{p}^{\cdot2}+\frac{1}{2}\|\vec{G}\|^{2}$.
Further note that\[
(\vec{p}\cdot\vec{G}^{*})\frac{1}{\frac{1}{2}|\vec{k}|^{2}+|\vec{k}|}(\vec{G}\cdot\vec{p})
\ = \ 
g^{2}|\vec{p}|^{2} \bigg( 2\pi^{2} - \frac{8\pi}{3} \bigg) \: \ln\bigg(\frac{\Lambda+2}{\sigma+2}\bigg)
\]
and in particular does not depend on the choice of the polarization vectors~$\vec{\varepsilon}$.

The quantity $\vec{G}^{\cdot\vee2*}S^{-1}\vec{G}^{\cdot\vee2}$ does
not depend on the choice of the vectors $\vec{\varepsilon}$ either
since\[
\vec{G}^{\cdot\vee2*}S^{-1}\vec{G}^{\cdot\vee2}=\sum_{\mu,\nu=\pm}\int\frac{|\vec{\varepsilon}(\vec{k}_{1},\mu)\cdot\vec{\varepsilon}(\vec{k}_{2},\nu)|^{2}}{\sqrt{|\vec{k}_{1}||\vec{k}_{2}|}S(\vec{k}_{1},\vec{k}_{2})}d^{3}k_{1}d^{3}k_{2}\]
and with $P_{\vec{u}}$ is the orthogonal projection on $\vec{u}$
in $\mathbb{R}^{3}$,\begin{align*}
\sum_{\mu,\nu=\pm}|\vec{\varepsilon}(\vec{k}_{1},\mu)\cdot\vec{\varepsilon}(\vec{k}_{2},\nu)|^{2} & =\sum_{\mu,\nu=\pm}\Tr_{\mathbb{R}^{3}}[P_{\vec{\varepsilon}(\vec{k}_{1},\mu)}P_{\vec{\varepsilon}(\vec{k}_{2},\nu)}]\\
 & =\Tr_{\mathbb{R}^{3}}[P_{\vec{k}_{1}}^{\perp}P_{\vec{k}_{2}}^{\perp}]\\
 & =1+\bigg(\frac{\vec{k}_{1}}{|\vec{k}_{1}|}\cdot\frac{\vec{k}_{2}}{|\vec{k}_{2}|}\bigg)^{2}\,.\end{align*}
\end{rem}
\begin{proof}[Proof of Theorem~\ref{thm:perturbative}]
Let\[
F:(g,\vec{p},f,\hat{r})\mapsto\partial_{f,\hat{r}}\hat{\mathcal{E}}_{g,\vec{p}}(f,\hat{r})\]
and $\left(\begin{smallmatrix}f\\
\hat{r}\end{smallmatrix}\right)(g,\vec{p}):=\left(\begin{smallmatrix}f(g,\vec{p})\\
\hat{r}(g,\vec{p})\end{smallmatrix}\right)$ such that \begin{equation}
F(g,\vec{p},\left(\begin{smallmatrix}f\\
\hat{r}\end{smallmatrix}\right)(g,\vec{p}))=0\,,\label{eq:F(gPfr)egal0}\end{equation}
then a derivation of Equation~(\ref{eq:F(gPfr)egal0}) with respect
to $(f,\hat{r})$ brings\[
\partial_{g,\vec{p}}\left(\begin{smallmatrix}f\\
\hat{r}\end{smallmatrix}\right)(0_{g,\vec{p}})=-\big[\partial_{f,\hat{r}}F(0_{g,\vec{p}},0_{f,\hat{r}})\big]^{-1}\partial_{g,\vec{p}}F(0_{g,\vec{p}},0_{f,\hat{r}})\,.\]
The term which is independent of $(g,\vec{p})$ and quadratic in $\left(\begin{smallmatrix}f\\
\hat{r}\end{smallmatrix}\right)$ in the energy is \[
\frac{1}{2}\{\Tr[\hat{r}S\hat{r}]+f^{*}(|\vec{k}|^{2}+2|\vec{k}|)f\}\]
 thus, in $(0_{g,\vec{p}},0_{f,\hat{r}})$, \[
\partial_{f,\hat{r}}F=\left(\begin{array}{cc}
|\vec{k}|^{2}+2|\vec{k}| & 0\\
0 & S\end{array}\right)\,.\]
To compute $\partial_{g,\vec{p}}F$ in $0$, observe that no part in the energy is linear in $(g,\vec{p})$ and linear
in $(f,\hat{r})$. Thus~$\partial_{g,\vec{p}}F(0_{g,\vec{p}},0_{f,\hat{r}})=0$
and we get\[
\partial_{g,\vec{p}}f(0_{g,\vec{p}})=0\,.\]
Differentiating a second time Equation~(\ref{eq:F(gPfr)egal0}) brings\[
0=\partial_{g,\vec{p}}^{2}F+2\partial_{f,\hat{r}}\partial_{g,\vec{p}}F\circ\partial_{g,\vec{p}}\left(\begin{smallmatrix}f\\
\hat{r}\end{smallmatrix}\right)+\partial_{f,\hat{r}}F\circ\partial_{g,\vec{p}}^{2}\left(\begin{smallmatrix}f\\
\hat{r}\end{smallmatrix}\right)+\partial_{f,\hat{r}}^{2}F(\partial_{g,\vec{p}}\left(\begin{smallmatrix}f\\
\hat{r}\end{smallmatrix}\right),\partial_{g,\vec{p}}\left(\begin{smallmatrix}f\\
\hat{r}\end{smallmatrix}\right))\,.\]
Since $\partial_{g,\vec{p}}\left(\begin{smallmatrix}f\\
\hat{r}\end{smallmatrix}\right)(0_{g,\vec{p}})=0$, it follows that\begin{align*}
\partial_{g,\vec{p}}^{2}\left(\begin{smallmatrix}f\\
\hat{r}\end{smallmatrix}\right)(0_{g,\vec{p}}) & =-[\partial_{f,\hat{r}}F(0_{g,\vec{p}},0_{f,\hat{r}})]^{-1}\partial_{g,\vec{p}}^{2}F(0_{g,\vec{p}},0_{f,\hat{r}})\,.\end{align*}
The part of the energy which is quadratic in $(g,\vec{p})$ and linear
in $(f,\hat{r})$ is $-2\rRe(f^{*}\vec{G})\cdot\vec{p}+\rRe\langle\hat{r}\vec{G};\vec{G}\rangle$,
it follows that, in~$(0_{g,\vec{p}},0_{f,\hat{r}})$, \begin{align*}
\partial_{g,\vec{p}}^{2}F & =2\left(\begin{array}{c}
\left(\begin{array}{cc}
1 & 0\end{array}\right)\vee\left(\begin{array}{cc}
0 & -2\partial_{g}\vec{G}\end{array}\right)\\
\left(\begin{array}{cc}
\partial_{g}\vec{G} & 0\end{array}\right).\vee\left(\begin{array}{cc}
\partial_{g}\vec{G} & 0\end{array}\right)\end{array}\right)\,,\end{align*}
which gives in $0_{g,\vec{p}}$\[
\partial_{g,\vec{p}}^{2}\left(\begin{smallmatrix}f\\
\hat{r}\end{smallmatrix}\right)=2\left(\begin{array}{c}
\left(\begin{array}{cc}
1 & 0\end{array}\right)\vee\left(\begin{array}{cc}
0 & \frac{\partial_{g}\vec{G}}{\frac{1}{2}|\vec{k}|^{2}+|\vec{k}|}\end{array}\right)\\
-S^{-1}\left(\begin{array}{cc}
\partial_{g}\vec{G} & 0\end{array}\right)\cdot\vee\left(\begin{array}{cc}
\partial_{g}\vec{G} & 0\end{array}\right)\end{array}\right)\,.\]
Hence the expansion of~$\left(\begin{smallmatrix}f\\
\hat{r}\end{smallmatrix}\right)$ up to order~$2$.

We can thus express the energy around $0_{g,\vec{p}}$ modulo error
terms in~$\mathcal{O}(\|(g,\vec{p})\|^{5})$\begin{align*}
\min_{f,\hat{r}} & \,\hat{\mathcal{E}}_{g,\vec{p}}(f,\hat{r})-\hat{\mathcal{E}}_{g,\vec{p}}(0,0)\\
 & \equiv\tfrac{1}{2}\big\{(\Tr[\hat{r}^{2}\vec{k}]+f^{*}\vec{k}f+2\rRe(f^{*}\vec{G})-\vec{p})^{\cdot2}+\Tr[\hat{r}\vec{k}\cdot\hat{r}\vec{k}]+\Tr[|\vec{k}|^{2}\hat{r}^{2}]\\
 & \phantom{=}+2\rRe\langle\hat{r}(\vec{G}+\vec{k}f);(\vec{G}+\vec{k}f)\rangle+\|\vec{G}\|^{2}+f^{*}|\vec{k}|^{2}f\big\}\\
 & \phantom{=}+\Tr[\hat{r}^{2}|\vec{k}|]+f^{*}|\vec{k}|f-\hat{\mathcal{E}}_{g,\vec{p}}(0,0)\allowdisplaybreaks\\
 & \equiv-2\rRe(f^{*}\vec{G})\cdot\vec{p}+\tfrac{1}{2}\Tr[\hat{r}S\hat{r}]+\rRe\langle\hat{r}\vec{G};\vec{G}\rangle+f^{*}(\frac{1}{2}|\vec{k}|^{2}+|\vec{k}|)f\\
 & \equiv-2\rRe(f^{*}\vec{p}\cdot\vec{G})+\tfrac{1}{2}\Tr[\hat{r}S\hat{r}]+\rRe\langle\hat{r}\vec{G};\vec{G}\rangle+f^{*}(\frac{1}{2}|\vec{k}|^{2}+|\vec{k}|)f\\
 & \equiv-2\frac{(\vec{p}\cdot\vec{G})^{*}(\vec{p}\cdot\vec{G})}{\frac{1}{2}|\vec{k}|^{2}+|\vec{k}|}+\frac{1}{2}\vec{G}^{\cdot\vee2*}S^{-1}\vec{G}^{\cdot\vee2}-\vec{G}^{\cdot\vee2*}S^{-1}\vec{G}^{\cdot\vee2}+\frac{(\vec{p}\cdot\vec{G})^{*}(\vec{p}\cdot\vec{G})}{\frac{1}{2}|\vec{k}|^{2}+|\vec{k}|}\end{align*}
which completes the proof. 
\end{proof}

\section{\label{sec:Lagrange-equations}Lagrange Equations}

This section formulates the results of Section~\ref{sec:Existence-and-uniqueness}
in terms of $\gamma$ and $\alpha$ subject to the constraints $\gamma+\gamma^{2}=(\alpha^{*}\otimes\mathbf{1}_{\mathcal{Z}})(\mathbf{1}_{\mathcal{Z}}\otimes\alpha)$,
without reference to the parametrization of $\gamma$ and $\alpha$
in terms of $\hat{r}$.

Suppose $f\in\mathcal{Z}$, $\alpha\in\mathcal{Z}^{\vee2}$, $\gamma\in\mathcal{L}^{1}(\mathcal{Z})$,
$\lambda\in\mathcal{B}(\mathcal{Z})=\mathcal{B}$ and $\vec{u}\in\mathbb{R}^{3}$.
Let $\mathcal{A}(\lambda)=\frac{1}{2}\vec{k}\cdot\!\vee\vec{k}+\lambda\vee\mathbf{1}$
and $\mathcal{G}(\gamma)=\gamma+\gamma^{2}$.
\begin{thm}
\label{pro:Lagrange_Equations}Suppose $(f,\gamma,\alpha)$ is a
minimum of the energy functional $\mathcal{E}$ such that $\|\gamma\|_{\mathcal{B}(\mathcal{Z})}<\frac{1}{2}$.
Then there is a unique $(\lambda,\vec{u})$ such that \textup{$(f,\gamma,\alpha,\lambda,\vec{u})$}
satisfies the following equations, equivalent to Lagrange equations
\begin{align}
M(\gamma,\vec{u})f & =-(\vec{k}(\gamma+\frac{1}{2}\mathbf{1})-\vec{u})\cdot\vec{G}-\vec{k}\cdot\!\vee(\vec{G}+\vec{k}f)^{*}\alpha\label{eq:f_implicit}\\
\mathcal{A}(\lambda)\alpha & =-\frac{1}{2}(\vec{G}+\vec{k}f)^{\cdot\vee2}\label{eq:alpha_implicit}\\
\gamma & =\mathcal{G}^{-1}((\alpha^{*}\otimes\mathbf{1}_{\mathcal{Z}})(\mathbf{1}_{\mathcal{Z}}\otimes\alpha))\label{eq:gamma_explicit}\\
\lambda & =\int_{0}^{\infty}e^{-t(\frac{1}{2}+\gamma)}(M(\gamma,\vec{u})+(\vec{G}+\vec{k}f)\cdot(\vec{G}+\vec{k}f)^{*})e^{-t(\frac{1}{2}+\gamma)}dt\label{eq:lambda_explicit}\\
\vec{u} & =\vec{p}-\Tr[\gamma\vec{k}]-f^{*}\vec{k}f-2\rRe(f^{*}\vec{G})\label{eq:u}\end{align}
with $M(\gamma,\vec{u})=\frac{1}{2}|\vec{k}|^{2}+|\vec{k}|-\vec{k}\cdot\vec{u}+\vec{k}\cdot\gamma\vec{k}$.

Assuming $|\vec{p}|<\frac{1}{2}$, sufficient conditions such that
$M(\gamma,\vec{u})$ and $\mathcal{A}(\lambda)$ are invertible operators
are $|\vec{u}|<1/2$, $\gamma\geq0$ and $\|\lambda-(|\vec{k}|^{2}/2+|\vec{k}|-\vec{p}\cdot\vec{k})\|_{\mathcal{B}}<\sigma/2$.
Equations~(\ref{eq:f_implicit}) to (\ref{eq:u}) then form a system
of coupled explicit equations.\end{thm}
\begin{rem}
To prove that Equations~(\ref{eq:f_implicit}) to~(\ref{eq:u})
admit a solution we use here the result of existence of a minimizer
proved in Section~\ref{sec:Existence-and-uniqueness}. It can also
be proved directly by a fixed point argument by defining the applications\begin{align*}
\Psi_{f}(f,\alpha,\gamma,\vec{u}) & =-M(\gamma,\vec{u})^{-1}(\vec{k}(\gamma+\frac{1}{2}\mathbf{1})-\vec{u})\cdot\vec{G}-\vec{k}\cdot\!\vee(\vec{G}+\vec{k}f)^{*}\alpha\\
\Psi_{\alpha}(f,\lambda) & =-\mathcal{A}(\lambda)^{-1}\frac{1}{2}(\vec{G}+\vec{k}f)^{\cdot\vee2}\\
\Psi_{\gamma}(\alpha) & =\mathcal{G}^{-1}((\alpha^{*}\otimes\mathbf{1}_{\mathcal{Z}})(\mathbf{1}_{\mathcal{Z}}\otimes\alpha))\\
\Psi_{\lambda}(f,\gamma,\vec{u}) & =\int_{0}^{\infty}e^{-t(\frac{1}{2}+\gamma)}(M(\gamma,\vec{u})+(\vec{G}+\vec{k}f)\cdot(\vec{G}+\vec{k}f)^{*})e^{-t(\frac{1}{2}+\gamma)}dt\\
\Psi_{\vec{u}}(f,\gamma) & =\vec{p}-\Tr[\gamma\vec{k}]-f^{*}\vec{k}f-2\rRe(f^{*}\vec{G})\end{align*}
defined on balls of centers centrers $0$, $0$, $0$, $\frac{1}{2}|\vec{k}|^{2}+|\vec{k}|-\vec{k}.\vec{p}$
and $\vec{p}$ and proving that the application \begin{multline*}
\Psi_{(f,\lambda)}(f,\lambda)=\left(\Psi_{f}\left[f,\Psi_{\alpha}\left\{ f,\lambda\right\} ,\Psi_{\gamma}\left\{ \Psi_{\alpha}(f,\lambda)\right\} ,\Psi_{\vec{u}}\left\{ f,\Psi_{\gamma}(\Psi_{\alpha}\left[f,\lambda\right])\right\} \right],\right.\\
\left.\Psi_{\lambda}\left[f,\Psi_{\gamma}\left\{ \Psi_{\alpha}(f,\lambda)\right\} ,\Psi_{\vec{u}}\left\{ f,\Psi_{\gamma}(\Psi_{\alpha}\left[f,\lambda\right])\right\} \right]\right)\end{multline*}
is a contraction for a convenient choice of the radiuses and a sufficiently
small coupling constant~$g$. Note that it is then convenient to
consider the norm of $L^{2}(S_{\sigma,\Lambda}\times\mathbb{Z}_{2},|\vec{k}|^{2})$
for $f$.\end{rem}
\begin{proof}[Proof of~Theorem~\ref{pro:Lagrange_Equations}]
Indeed, set $\vec{u}=\vec{p}-\Tr[\gamma\vec{k}]-f^{*}\vec{k}f-2\rRe(f^{*}\vec{G})$
and define the partial derivatives as $\partial_{f^{*}}\mathcal{E}(f,\gamma,\alpha)\in\mathcal{Z}$,
$\partial_{\alpha^{*}}\mathcal{E}(f,\gamma,\alpha)\in\mathcal{Z}^{\vee2}$
and $\partial_{\gamma}\mathcal{E}(f,\gamma,\alpha)\in\mathcal{B}(\mathcal{Z})\cong\mathcal{L}^{1}(\mathcal{Z})^{\prime}$
such that\begin{align*}
\mathcal{E} & (f+\delta f,\gamma+\delta\gamma,\alpha+\delta\alpha)-\mathcal{E}(f,\gamma,\alpha)\\
 & =2\rRe(\delta\! f^{*}\,\partial_{f^{*}}\mathcal{E}(f,\gamma,\alpha))+2\rRe(\delta\alpha^{*}\,\partial_{\alpha^{*}}\mathcal{E}(f,\gamma,\alpha))\\
 & \quad+\Tr[\delta\gamma\,\partial_{\gamma}\mathcal{E}(f,\gamma,\alpha)]+o(\|(\delta f,\delta\gamma,\delta\alpha)\|_{\mathcal{Z}\times\mathcal{L}^{1}(\mathcal{Z})\times\mathcal{Z}^{\vee2}})\,.\end{align*}
Recall the energy functional is given by Equation~(\ref{eq:Energy})
and this yields\begin{align*}
\partial_{f^{*}}\mathcal{E}(f,\gamma,\alpha) & =\frac{1}{2}\big\{2(\vec{k}f+\vec{G})\cdot(\Tr[\gamma\vec{k}]+f^{*}\vec{k}f+2\rRe(f^{*}\vec{G})-\vec{p})\\
 & \phantom{=}+2\vec{k}\cdot\!\vee(\vec{G}+\vec{k}f)^{*}\alpha+\vec{k}\cdot(2\gamma+\mathbf{1})(\vec{G}+\vec{k}f)\big\}+|\vec{k}|f\\
 & =-(\vec{k}f+\vec{G})\cdot\vec{u}+\vec{k}\cdot\vee(\vec{G}+\vec{k}f)^{*}\alpha+\vec{k}\cdot(\gamma+\frac{1}{2}\mathbf{1})(\vec{G}+\vec{k}f)+|\vec{k}|f\\
 & =M(\gamma,\vec{u})f+(\vec{k}(\gamma+\frac{1}{2}\mathbf{1})-\vec{u})\cdot\vec{G}+\vec{k}\cdot\!\vee(\vec{G}+\vec{k}f)^{*}\alpha\,,\allowdisplaybreaks\\
\partial_{\alpha^{*}}\mathcal{E}(f,\gamma,\alpha) & =\frac{1}{2}(\vec{k}\cdot\!\otimes\vec{k})\alpha+\frac{1}{2}(\vec{G}+\vec{k}f)^{\cdot\vee2}\,,\allowdisplaybreaks\\
\partial_{\gamma}\mathcal{E}(f,\gamma,\alpha) & =\frac{1}{2}\big\{2\vec{k}\cdot(\Tr[\gamma\vec{k}]+f^{*}\vec{k}f+2\rRe(f^{*}\vec{G})-\vec{p})\\
 & \phantom{=}+2\vec{k}\cdot\gamma\vec{k}+|\vec{k}|^{2}+2(\vec{G}+\vec{k}f)\cdot(\vec{G}+\vec{k}f)^{*}\big\}+|\vec{k}|\\
 & =M(\gamma,\vec{u})+(\vec{G}+\vec{k}f)\cdot(\vec{G}+\vec{k}f)^{*}\,.\end{align*}
The constraint given by Equation~(\ref{eq:constraint}) can be
expressed as \begin{equation}
\mathcal{C}(f,\gamma,\alpha)=0\label{eq:constraint-C}\end{equation}
 with \begin{align*}
\mathcal{C}:\mathcal{Z}\times\mathcal{L}^{1}(\mathcal{Z})\times\mathcal{Z}^{\vee2} & \to\mathcal{L}^{1}(\mathcal{Z})\\
(f,\gamma,\alpha) & \mapsto\gamma+\gamma^{2}-(\alpha^{*}\otimes\mathbf{1}_{\mathcal{Z}})(\mathbf{1}_{\mathcal{Z}}\otimes\alpha)\,.\end{align*}
Equation~(\ref{eq:constraint-C}) is equivalent to Equation~(\ref{eq:gamma_explicit}).
The application~$\mathcal{C}$ has a differential $D\mathcal{C}(f,\gamma,\alpha):\mathcal{Z}\times\mathcal{L}^{1}(\mathcal{Z})\times\mathcal{Z}^{\vee2}\to\mathcal{L}^{1}(\mathcal{Z})$
such that\begin{multline*}
D\mathcal{C}(f,\gamma,\alpha)(\delta\! f,\delta\gamma,\delta\alpha)\\
=\delta\gamma+\delta\gamma\:\gamma+\gamma\,\delta\gamma-(\delta\alpha^{*}\otimes\mathbf{1}_{\mathcal{Z}})(\mathbf{1}_{\mathcal{Z}}\otimes\alpha)-(\alpha^{*}\otimes\mathbf{1}_{\mathcal{Z}})(\mathbf{1}_{\mathcal{Z}}\otimes\delta\alpha)\,.\end{multline*}
For $\|\gamma\|_{\mathcal{B}(\mathcal{Z})}<\frac{1}{2}$ the application
$D\mathcal{C}(f,\gamma,\alpha)$ is surjective. Indeed it is already
surjective on $\{0\}\times\mathcal{L}^{1}(\mathcal{Z})\times\{0\}$,
since, for every $\gamma^{\prime}\in\mathcal{L}^{1}(\mathcal{Z})$
the equation $\delta\gamma+\delta\gamma\,\gamma+\gamma\,\delta\gamma=\gamma^{\prime}$
with unknown~$\delta\gamma$ has at least one solution, see Proposition~\ref{pro:Sylvester_Equation}.
We can then apply the Lagrange multiplier rule (see for example the
book of Zeidler~\cite{MR1347692})  which tells us that there exists
a $\lambda\in\mathcal{B}(\mathcal{Z})$ such that \[
\forall(\delta\! f,\delta\alpha,\delta\gamma)\,,\quad D\mathcal{E}(f,\alpha,\gamma)(\delta\! f,\delta\alpha,\delta\gamma)+\Tr[D\mathcal{C}(f,\alpha,\gamma)(\delta\! f,\delta\alpha,\delta\gamma)\,\lambda]=0\,,\]
that is to say\begin{multline*}
2\rRe(\delta\! f^{*}\partial_{f^{*}}\mathcal{E}(f,\gamma,\alpha)+\delta\alpha^{*}\partial_{\alpha^{*}}\mathcal{E}(f,\gamma,\alpha))+\Tr[\partial_{\gamma}\mathcal{E}(f,\gamma,\alpha)\delta\gamma]\\
+\Tr[(\delta\gamma+\delta\gamma\,\gamma+\gamma\,\delta\gamma-(\delta\alpha^{*}\otimes\mathbf{1}_{\mathcal{Z}})(\mathbf{1}_{\mathcal{Z}}\otimes\alpha)-(\alpha^{*}\otimes\mathbf{1}_{\mathcal{Z}})(\mathbf{1}_{\mathcal{Z}}\otimes\delta\alpha))\lambda]=0\,.\end{multline*}
This is equivalent to Equations~(\ref{eq:f_implicit}), (\ref{eq:alpha_implicit})
and\begin{align}
\lambda(\frac{1}{2}+\gamma)+(\frac{1}{2}+\gamma)\lambda & =M(\gamma,\vec{u})+(\vec{G}+\vec{k}f)\cdot(\vec{G}+\vec{k}f)^{*}\label{eq:lambda_implicit}\end{align}
Using again Proposition~\ref{pro:Sylvester_Equation} we get that
Equation~(\ref{eq:lambda_implicit}) is equivalent to Equation~(\ref{eq:lambda_explicit}). 

For the invertibility of $\mathcal{A}(\lambda)$ note that\begin{align*}
\mathcal{A}(\lambda) & =\frac{1}{4}(\vec{k}\otimes\mathbf{1}+\mathbf{1}\otimes\vec{k})^{\cdot2}+(|\vec{k}|-\vec{k}.\vec{p}+\lambda-\frac{1}{2}|\vec{k}|^{2}-|\vec{k}|+\vec{k}\cdot\vec{p})\vee\mathbf{1}\\
 & \geq(\frac{\sigma}{2}-\lambda-(|\vec{k}|^{2}/2+|\vec{k}|-\vec{p}\cdot\vec{k})\|_{\mathcal{B}})\,\mathbf{1}\vee\mathbf{1}\,.\end{align*}
For $M(\gamma,\vec{u})$, $M(\gamma,\vec{u})=\frac{1}{2}|\vec{k}|^{2}+|\vec{k}|-\vec{k}\cdot\vec{u}+\vec{k}\cdot\gamma\vec{k}\geq\sigma/2$
if $\gamma\geq0$ and $|\vec{u}|<1/2$.
\end{proof}

Let us recall a well known expression for the solution of the Sylvester
or Lyapunov equation.
\begin{prop}
\label{pro:Sylvester_Equation}Let $A$ and $B$ be bounded self-adjoint
operators on a Hilbert space. Suppose $A\geq a\,\mathbf{1}$ with
$a>0$. Then the equation\[
AX+XA=B\]
for $X$ a bounded operator has a unique solution $\chi_{A}(B)=\int_{0}^{\infty}e^{-tA}Be^{-tA}dt$.

If $B$ a trace class operator then the solution $X$ is also trace
class.\end{prop}
\begin{proof}
Indeed, $\chi_{A}(B)$ is a solution because\begin{align*}
A\chi_{A}(B)+\chi_{A}(B)A & =\int_{0}^{\infty}e^{-tA}(AB+BA)e^{-tA}dt\\
 & =-\int_{0}^{\infty}\frac{d}{dt}(e^{-tA}Be^{-tA})dt=B\,.\end{align*}
Conversely, suppose that $AX+XA=B$, then \begin{align*}
\chi_{A}(B) & =\int_{0}^{\infty}e^{-tA}(AX+XA)e^{-tA}dt\\
 & =-\int_{0}^{\infty}\frac{d}{dt}(e^{-tA}Xe^{-tA})dt=X\,,\end{align*}
and thus any solution $X$ is equal to $\chi_{A}(B)$. Hence the solution is unique.
\end{proof}

\bibliographystyle{plain}
\bibliography{biblio}

\end{document}

%% file: Intro_BBT-1-13-01-04.tex
\secct{Introduction} 
\label{sec-I}
%
\subsection{The Hamiltonian} \label{subsec-I.1}
%
According to the \emph{Standard Model of Nonrelativistic Quantum
  Electrodynamics} \cite{MR1639713} the unitary time
evolution of a free nonrelativistic particle coupled to the quantized
radiation field is generated by the Hamiltonian
\begin{align} \label{eq-I.1} 
\tH_g \ := \ 
\frac{1}{2}\big( \tfrac{1}{i} \vnabla_x - \vAA(\vx) \big)^2 + \hf
\end{align}
acting on the Hilbert space $L^2(\RR_x^3; \fF_\sL)$ of
square-integrable functions with values in the photon Fock space
\begin{align} \label{eq-I.2} 
\fF_\sL \ := \  
\fF_+(\cZ_\sL) 
\ := \ 
\bigoplus_{n=0}^\infty \fF_+^{(n)}(\cZ_\sL) ,
\end{align}
where $\fF_+^{(0)}(\cZ_\sL) = \CC \cdot \Om$ is the
vacuum sector and the $n$-photon sector
$\fF_+^{(n)}(\cZ_\sL) = \cS(\cZ_\sL^{\otimes  n} )$ 
is the subspace of totally symmetric vectors on the $n$-fold
tensor product of the one-photon Hilbert space
\begin{align} \label{eq-I.2,1} 
\cZ_\sL \ = \  
\big\{ \vf \in L^2(S_{\sigma,\Lambda}; \CC \otimes \RR^3) \; \big|
\ \forall \vk \in S_{\sigma,\Lambda} \ a.e.: \quad \vk \cdot \vf(\vk) = 0 \big\}
\end{align}
of square-integrable, transversal vector fields which are supported in
the momentum shell 
\begin{align} \label{eq-I.2,2} 
S_{\sigma,\Lambda} \ := \ 
\big\{ \vk \in \RR^3 \; \big| \ \sigma \leq |\vk| \leq \Lambda \big\},
\end{align}
where $0 \leq \sigma < \Lambda < \infty$ are infrared and ultraviolet
cutoffs, respectively, reflecting our choice of gauge, namely, the
Coulomb gauge. It is convenient to fix real polarization vectors
$\veps_\pm(\vk) \in \RR^3$ such that $\{ \veps_+(\vk), \veps_-(\vk),
\tfrac{\vk}{|\vk|} \} \subseteq \RR^3$ form a right-handed orthonormal
basis (Dreibein) and replace \eqref{eq-I.2,1} by
\begin{align} \label{eq-I.3} 
\cZ_\sL  \ = \  L^2(S_{\sigma,\Lambda} \times \ZZ_2) ,
\end{align}
with the understanding that $\vf(\vk) = \veps_+ \, f(\vk,+) + \veps_-
\, f(\vk,-)$.

In \eqref{eq-I.1} the energy of the photon field is represented by
\begin{align} \label{eq-I.4} 
\hf \ = \ \int |k| \, a^*(k) \, a(k) \, dk ,  
\end{align}
where $\int f(k) dk := \sum_{\tau = \pm} \int_{S_{\sigma,\Lambda}} f(\vk,\tau) \,
d^3k$ and $\{a(k), a^*(k)\}_{k \in S_{\sigma,\Lambda} \times \ZZ_2}$ are the usual
boson creation and annihilation operators constituing a Fock
representation of the CCR on $\fF_\sL$, i.e.,
\begin{align} \label{eq-I.5} 
[ a(k) \, , \, a(k') ] \ = \ & [ a^*(k) \, , \, a^*(k') ] \ = \ 0,
\\ \label{eq-I.6} 
[ a(k) \, , \, a^*(k') ] \ = \ & \delta(k-k') \: \one , \quad
a(k) \Om \ = \ 0, 
\end{align}
for all $k, k' \in S_{\sigma,\Lambda} \times \ZZ_2$. The
magnetic vector potential $\vAA(\vx)$ is given by
\begin{align} \label{eq-I.7} 
\vAA(\vx) \ = \ 
 \, \int \vG(k) 
\Big( e^{-i \vk \cdot \vx} \, a^*(k) + e^{i \vk \cdot \vx} \, a(k) \Big) \: dk , 
\end{align}
with $k = (\vk,\tau) \in \RR^3 \times \ZZ_2$,
\begin{align} \label{eq-I.8} 
\vG(\vk,\tau) \ := \, g\, \veps_\tau(\vk) \: |\vk|^{-1/2} ,
\end{align}
and $g \in \RR$ being the coupling constant. In our units, the mass
of the particle and the speed of light equal one, so the coupling
constant is given as $g = \tfrac{1}{4\pi}
\sqrt{\alpha}$, with $\alpha \approx 1/137$ being Sommerfeld's fine
structure constant.

The Hamiltonian $\tH_g$ preserves (i.e., commutes with) the total
momentum operator $\vp = \tfrac{1}{i} \vnabla_x + \pf$ of the system,
where
\begin{align} \label{eq-I.9} 
\pf \ = \  \int \vk \, a^*(k) \, a(k) \, dk 
\end{align}
is the photon field momentum. This fact allows us to eliminate the
particle degree of freedom. More specifically, introducing the unitary
\begin{align} \label{eq-I.10} 
\UU \; : L^2(\RR_x^3; \fF_\sL) \ \to \ L^2(\RR_p^3; \fF_\sL) \, , \quad
\big( \UU \Psi \big)(\vp) 
\ := \ 
\int e^{-i \vx \cdot (\vp - \pf)} \Psi(\vx) \: \frac{d^3x}{(2\pi)^{3/2}} ,
\end{align}
one finds that
\begin{align} \label{eq-I.11} 
\UU \, \tH_g \, \UU^* 
\ = \ 
\int^{\oplus} H_{g,\vp} \: d^3p ,
\end{align}
where
\begin{align} \label{eq-I.12} 
H_{g,\vp} \ = \ \frac{1}{2}\big( \pf + \vAA(\vO) - \vp \big)^2 \: + \: \hf 
\end{align}
is a selfadjoint operator on $\dom(H_{0,\vO})$, the natural domain 
of $H_{0,\vO} = \frac{1}{2} \pf^{\, 2} + \hf$.

\subsection{Ground State Energy} 
\label{subsec-I.2}
%
Due to \eqref{eq-I.11}, all spectral properties of $\tH_g$ are
obtained from those of $\{ H_{g,\vp} \}_{\vp \in \RR^3}$. Of
particular physical interest is the mass shell for fixed total
momentum $\vp \in \RR^3$, coupling constant $g \geq 0$, and infrared 
and ultraviolet cutoffs $0 \leq \sigma < \Lambda < \infty$, i.e., the
value of the ground state energy
\begin{align} \label{eq-I.13} 
E_\gs(g,\vp, \sigma, \Lambda) 
\ := \ 
\inf \sigma[H_{g,\vp}] 
\ \geq \ 0 
\end{align}
and the corresponding ground states (or approximate ground states).

We express the ground state energy in terms of density matrices with
finite energy expectation value and accordingly introduce
\begin{align} \label{eq-I.14} 
\tDM \ := \ \Big\{ \rho \in \cL^1(\fF) \; \Big| \
\rho \geq 0, \ \ \Tr_\fF[\rho] = 1, \ \ 
\rho \, H_{0,\vO} , \, H_{0,\vO} \, \rho \, \in \cL^1(\fF_\sL) \Big\},
\end{align}
so that the Rayleigh-Ritz principle appears in the form
\begin{align} \label{eq-I.15} 
E_\gs(g,\vp) \ = \ &
\inf\Big\{ \Tr_\fF\big[ \rho \, H_{g,\vp} \big] 
\; \Big| \ \rho \in \tDM \Big\} .
\end{align}
Note that $\Tr_\fF[ \rho \, H_{g,\vp} ] = \Tr_\fF[ \rho^{1-\beta} \,
  H_{g,\vp} \, \rho^\beta]$, for all $0 \leq \beta \leq 1$, due to our
assumption $\rho H_{0,\vO}, H_{0,\vO} \rho \in \cL^1(\fF_\sL)$.

It is not difficult to see that the ground state energy is already
obtained as an infimum over all density matrices
\begin{align} \label{eq-I.16} 
\DM \ := \ \Big\{ \rho \in \tDM \; \Big| \
\rho \, \nf , \, \nf \, \rho \, \in \cL^1(\fF_\sL) \Big\}
\end{align}
of finite photon number expectation value, where
\begin{align} \label{eq-I.17} 
\nf \ = \  \int a^*(k) \, a(k) \, dk 
\end{align}
is the photon number operator. Indeed, if $\sigma >0$ then
\begin{align} \label{eq-I.19} 
H_{g,\vp} \ \geq \ \hf \ \geq \ \sigma \, \nf ,
\end{align}
and $\DM = \tDM$ is automatic. Furthermore, if $\sigma = 0$ then it is
not hard to see \cite{MR1639713} that $E_\gs(g,\vp, 0,
\Lambda) = \lim_{\sigma \searrow 0} E_\gs(g,\vp, \sigma, \Lambda)$,
by using the standard relative bound
\begin{align} \label{eq-I.19,1} 
\big\|  \vAA_{< \sigma} (\vO) \, \psi \big\| 
\ \leq \ 
\cO(\sigma) \, \big\|  ( H_{f, < \sigma} + 1 )^{1/2} \, \psi \big\| ,
\end{align}
where $\vAA_{< \sigma} (\vO)$ and $H_{f, < \sigma}$ are the quantized
magnetic vector potential and field energy, respectively, for momenta
below $\sigma$. So, for all $0 \leq \sigma < \Lambda < \infty$, we have
that
\begin{align} \label{eq-I.20} 
E_\gs(g,\vp,\sigma, \Lambda) 
\ = \ &
\inf\Big\{ \Tr_\fF\big[ \rho \, H_{g,\vp}(\sigma, \Lambda) \big] 
\; \Big| \ \rho \in \DM \Big\},
\end{align}
indeed. If the infimum \eqref{eq-I.20} is attained at
$\rho_\gs(g,\vp,\sigma, \Lambda) \in \DM$ then we call
$\rho_\gs(g,\vp,\sigma, \Lambda)$ a ground state of $H_{g,\vp}(\sigma,
\Lambda)$. 

Since $\DM$ is convex, we may restrict the density matrices in
\eqref{eq-I.20} to vary only over pure density matrices,
\begin{align} \label{eq-I.21} 
E_\gs(g,\vp,\sigma, \Lambda) 
\ = \ 
\inf\Big\{ \Tr_\fF\big[ \rho \, H_{g,\vp}(\sigma, \Lambda) \big] 
\; \Big| \ \rho \in \pDM \Big\} ,
\end{align}
where \emph{pure} density matrices are those of rank one,
\begin{align} \label{eq-I.22} 
\tpDM \ := \ \Big\{ \rho \in \tDM \; \Big| \ 
\exists \Psi \in \fF_\sL, \,
\|\Psi\| = 1 : \ \ \rho = |\Psi\ra\la\Psi| \Big\} ,
\end{align}
and
\begin{align} \label{eq-I.23} 
\pDM \ := \ \DM \: \cap \: \tpDM .
\end{align}
Another class of states that play an important role in our work is
the set of \emph{centered} density matrices, 
\begin{align} \label{eq-I.23,1} 
\cDM \ := \ \Big\{ \rho \in \DM \; \Big| \ \forall f \in \cZ :
\ \ \Tr_\fF\big[ \rho \, a^*(f) \big] = 0 \Big\} .
\end{align}
%

\subsection{Bogolubov-Hartree-Fock Energy} 
\label{subsec-I.3}
%
The determination of $E_\gs(g,\vp)$ and the corresponding ground state
$\rho_\gs(g,\vp) \in \DM$ (provided the infimum is attained) is a
difficult task. In this paper we rather study approximations to
$E_\gs(g,\vp)$ and $\rho_\gs(g,\vp)$ that we borrow from the quantum
mechanics of atoms and molecules, namely, the Bogolubov-Hartree-Fock
(BHF) approximation. We define the BHF energy as
\begin{align} \label{eq-I.24} 
E_\BHF(g,\vp,\sigma, \Lambda) 
\ = \ &
\inf\Big\{ \Tr_\fF\big[ \rho \, H_{g,\vp}(\sigma, \Lambda) \big] 
\; \Big| \ \rho \in \QF \Big\},
\end{align}
with corresponding BHF ground state(s) $\rho_\BHF(g,\vp, \sigma,
\Lambda) \in \QF$, determined by
\begin{align} \label{eq-I.25} 
\Tr_\fF\big[ \rho_\BHF(g,\vp, \sigma, \Lambda) \; 
H_{g,\vp}(\sigma, \Lambda) \big] 
\ = \
E_\BHF(g,\vp,\sigma, \Lambda) ,
\end{align}
where
\begin{align} \label{eq-I.26} 
\QF \ := \ 
\Big\{ \rho \in \DM \; \Big| \ \text{$\rho$ is quasifree} \Big\}
\ \subseteq \
\DM
\end{align}
denotes the subset of quasifree density matrices. A density matrix
$\rho \in \DM$ is called \emph{quasifree}, if there exist $f_\rho \in
\cZ_\sL$ and a positive, self-adjoint operator $h_\rho =
h_\rho^* \geq 0$ on $\cZ_\sL$ such that
\begin{align} \label{eq-I.27} 
\big\la W(\sqrt{2}f/i) \big\ra_\rho := 
\Tr_\fF \big[ \rho \, W(\sqrt{2}f/i) \big] 
= 
\exp\Big[ 2i \, \rIm\la f_\rho | f \ra 
-  \big\la f \big| (1+h_\rho) f \big\ra \Big],
\end{align}
for all $f \in \cZ_\sL$, where 
\begin{align} \label{eq-I.28} 
W(f) \ := \ 
\exp\big[ i \Phi(f) \big] 
\ := \ 
\exp\big[ \tfrac{i}{\sqrt{2\,}} \big( a^*(f) + a(f) \big) \big] 
\end{align}
denotes the Weyl operator corresponding to $f$ and we write
expectation values w.r.t.\ the density matrix $\rho$ as $\la \cdot
\ra_\rho$. 

There are several important facts about quasifree density matrices,
which do not hold true for general density matrices in $\DM$. See,
e.g., \cite{MR1297873, MR2238912, MR545651, MR611508}. The first such
fact is that if $\rho \in \QF$ is a quasifree density matrix then so
is $W(g)^* \rho W(g) \in \QF$, for any $g \in \cZ_\sL$,
as follows from the Weyl commutation relations
\begin{align} \label{eq-I.28,1} 
\forall \, f,g \in \cZ_\sL: \qquad 
W(f) \, W(g) 
\ = \ 
e^{-\frac{i}{2} \rIm \la f | g \ra} \, W(f+g).
\end{align}
Choosing $g := -i\sqrt{2}f_{\rho}$, we find that
$W(-i\sqrt{2}f_{\rho})^* \, \rho \, W(-i\sqrt{2}f_{\rho})$ is a
centered quasifree density matrix, i.e.,
\begin{align} \label{eq-I.28,2} 
W(\sqrt{2}f_\rho/i)^* \, \rho \, W(\sqrt{2}f_\rho/i) 
\ \in \
\cQF \ := \ \QF \cap \cDM .
\end{align}
Next, we formulate a characterization of centered quasifree density
matrices.
\begin{lemma} \label{lem-1} 
Let $\rho \in \cDM$ be a centered density matrix and denote 
$\la A \ra_\rho := \Tr_\fF\{ \rho A \}$. Then $(i) \Leftrightarrow (ii)
\Leftrightarrow (iii)$, where
\begin{itemize}
\item[(i)] $\rho \in \cQF$ is centered and quasifree;

\item[(ii)] All odd correlation functions and all even trunctated
  correlation functions of $\rho$ vanish, i.e., for all $N \in \NN$
  and $\vphi_1, \ldots, \vphi_{2N} \in \cZ_\sL$, let
  either $b_n := a^*(\vphi_n)$ or $b_n := a(\vphi_n)$, for all $1 \leq
  n \leq 2N$. Then $\la b_1 \cdots b_{2N-1} \ra_\rho = 0$ and
\begin{align} \label{eq-I.29} 
\big\la b_1 \, b_2 \cdots b_{2N} \big\ra_\rho
\ = \ &
\sum_{\pi \in \fP_{2N}} \big\la b_{\pi(1)} \, b_{\pi(2)} \big\ra_\rho \,
\cdots \, \big\la b_{\pi(2N-1)} \, b_{\pi(2N)} \big\ra_\rho ,
\end{align}
where $\fP_{2N}$ denotes the set of pairings, i.e., the set of all
permutations $\pi \in \fS_{2N}$ of $2N$ elements such that $\pi(2n-1)
< \pi(2n+1)$ and $\pi(2n-1) < \pi(2n)$, for all $1 \leq n \leq N-1$
and $1 \leq n \leq N$, respectively.

\item[(iii)] There exist two commuting quadratic, semibounded
  Hamiltonians
\begin{align} \label{eq-I.30} 
H \ = \ &
\sum_{i,j} \Big\{ B_{i,j} \, a^*(\psi_i) \, a(\psi_j) 
\: + \: C_{i,j} a^*(\psi_i) \, a^*(\psi_j) 
\: + \: \ol{C_{i,j}} a(\psi_i) \, a(\psi_j) \Big\} 
\\[1ex] \label{eq-I.31} 
H' \ = \ &
\sum_{i,j} \Big\{ B_{i,j}' \, a^*(\psi_i) \, a(\psi_j) 
\: + \: C_{i,j}' a^*(\psi_i) \, a^*(\psi_j) 
\: + \: \ol{C_{i,j}'} a(\psi_i) \, a(\psi_j) \Big\} 
\end{align}
with $B = B^* \geq 0$, $C = C^T \in \cL^2(\cZ_\sL)$,
where $\{\psi_i\}_{i \in \NN} \subseteq \cZ_\sL$ is an
orthonormal basis, such that $\exp(-H -\beta H')$ is trace class, for
all $\beta < \infty$, and
\begin{align} \label{eq-I.32} 
\la A \ra_\rho \ = \ &
\lim_{\beta \to \infty} \bigg\{ 
\frac{ \Tr_\fF[ A \, \exp(-H -\beta H') ]}{\Tr_\fF[ \exp(-H -\beta H') ]}
\bigg\} ,
\end{align}
for all $A \in \cB(\fF_\sL)$.
\end{itemize}
\end{lemma}
Eq.~\eqref{eq-I.28,2} and the vanishing (ii) of the truncated
correlation functions of a centered quasifree state imply that any
quasifree state $\rho \in \QF$ is completely determined by its
one-point function $\la a(\vphi) \ra_\rho$ and its two-point function
(one-particle reduced density matrix)
\begin{align} \label{eq-I.33} 
\Gamma[\gamma_\rho, \talpha_\rho] \ := \ 
\begin{pmatrix}
\gamma_\rho & \talpha_\rho \\ \talpha_\rho^* & \one + \mathcal{J} \, \gamma_\rho \, \mathcal{J}
\end{pmatrix} 
\ \in \ 
\cB\big( \cZ_\sL \oplus \cZ_\sL ),
\end{align}
where the operators $\gamma_\rho, \talpha_\rho \in
\cB(\cZ_\sL)$ are defined as
\begin{align} \label{eq-I.34} 
\la \vphi, \: \gamma_\rho \, \psi \ra 
\ := \ 
\la a^*(\psi) \, a(\vphi) \ra_\rho 
\quad \text{and} \quad
\la \vphi, \: \talpha_\rho \, \psi \ra 
\ := \ 
\la a(\vphi) \, a(\mathcal{J} \psi) \ra_\rho ,
\end{align}
and $\mathcal{J}: \cZ_\sL \to \cZ_\sL$ is a conjugation. 
The positivity of the density matrix $\rho$
implies that $\Gamma[\gamma_\rho, \talpha_\rho] \geq 0$ and, in
particular, $\gamma_\rho \geq 0$, too. Moreover, the additional
finiteness of the particle number expectation value, which
distinguishes $\DM$ from $\tDM$, ensures that $\gamma_\rho \in
\cL^1(\cZ_\sL)$ is trace-class, namely,
\begin{align} \label{eq-I.35} 
\Tr_\cZ[ \gamma_\rho ] \ = \ \la \nf \ra_\rho \ < \ \infty , 
\end{align}
and that $\talpha_\rho \in \cL^2(\cZ_\sL)$ is
Hilbert-Schmidt. 

Similar to \eqref{eq-I.22}-\eqref{eq-I.23}, we introduce pure
quasifree density matrices,
\begin{align} \label{eq-I.35,1} 
\pQF \ := \ \QF \: \cap \: \tpDM .
\end{align}
A subset of $\pQF$ of special interest is given by
\emph{coherent states}, i.e., pure quasifree states of
the form $|W(-i\sqrt{2} f)\Om\ra\la W(-i\sqrt{2} f)\Om|$, which we collect in
\begin{align} \label{eq-I.36} 
\coherent \ := \ 
\big\{ |W(-i\sqrt{2} f)\Om\ra\la W(-i\sqrt{2} f)\Om| \ \big| \ 
f \in \cZ_\sL \big\} .
\end{align}
For these, $\gamma_\rho = \talpha_\rho = 0$.

Conversely, if $\gamma \in \cL_+^1(\cZ_\sL)$ is a
positive trace-class operator and $\tilde\alpha \in
\cL^2(\cZ_\sL)$ is a Hilbert-Schmidt operator such that
$\Gamma[\gamma, \talpha] \geq 0$ is positive then there exists a
unique centered quasifree density matrix $\rho \in \cQF$ such that
$\gamma = \gamma_\rho$ and $\alpha = \alpha_\rho$ are its one-particle
reduced density matrices.

Summarizing these two relations, the set $\QF$ of quasifree density
matrices is in one-to-one correspondence to the convex set
\begin{align} \label{eq-I.37} 
\rdm \ := \ 
\Big\{ (f, \gamma, \talpha) \in 
\cZ_\sL \oplus \cL_+^1(\cZ_\sL) \oplus \cL^2(\cZ_\sL) 
\ \Big| \ \ \Gamma[\gamma, \talpha] \geq 0 \Big\}.
\end{align}
Note that coherent states correspond to elements of $\rdm$ of the
form $(f,0,0)$. 

Next, we observe in accordance with \eqref{eq-I.37} that, if $\rho
\in \QF$ is quasifree then its energy expectation value $\la H_{g,\vp}
\ra_\rho$ is a functional of $(f_\rho, \gamma_\rho, \talpha_\rho)$, namely,
\begin{align} \label{eq-I.38} 
\big\la H_{g,\vp} \big\ra_\rho 
\ = \ 
\cE_{g,\vp}(f_\rho, \gamma_\rho, \talpha_\rho) ,
\end{align}
where
%
%
%
%
%
%
%

%% file: Bach_Breteaux_Tzaneteas-13-01-05.bbl
\begin{thebibliography}{10}

\bibitem{MR2326223}
Gr{{\'e}}goire Allaire.
\newblock {\em Numerical analysis and optimization}.
\newblock Numerical Mathematics and Scientific Computation. Oxford University
  Press, Oxford, 2007.
\newblock An introduction to mathematical modelling and numerical simulation,
  Translated from the French by Alan Craig.

\bibitem{MR1639713}
Volker Bach, J{{\"u}}rg Fr{{\"o}}hlich, and Israel~Michael Sigal.
\newblock Quantum electrodynamics of confined nonrelativistic particles.
\newblock {\em Adv. Math.}, 137(2):299--395, 1998.

\bibitem{MR1297873}
Volker Bach, Elliott~H. Lieb, and Jan~Philip Solovej.
\newblock Generalized {H}artree-{F}ock theory and the {H}ubbard model.
\newblock {\em J. Statist. Phys.}, 76(1-2):3--89, 1994.

\bibitem{MR0208930}
F.~A. Berezin.
\newblock {\em The method of second quantization}.
\newblock Translated from the Russian by Nobumichi Mugibayashi and Alan
  Jeffrey. Pure and Applied Physics, Vol. 24. Academic Press, New York, 1966.

\bibitem{MR545651}
Ola Bratteli and Derek~W. Robinson.
\newblock {\em Operator algebras and quantum statistical mechanics. {V}ol. 1}.
\newblock Springer-Verlag, New York, 1979.
\newblock $C^{\ast} $- and $W^{\ast} $-algebras, algebras, symmetry groups,
  decomposition of states, Texts and Monographs in Physics.

\bibitem{MR611508}
Ola Bratteli and Derek~W. Robinson.
\newblock {\em Operator algebras and quantum-statistical mechanics. {II}}.
\newblock Springer-Verlag, New York, 1981.
\newblock Equilibrium states. Models in quantum-statistical mechanics, Texts
  and Monographs in Physics.

\bibitem{BRETEAUX2012}
S{\'e}bastien Breteaux.
\newblock {Higher order terms for the quantum evolution of a Wick observable
  within the Hepp method}.
\newblock {\em {Cubo (Temuco)}}, 14:91--109, 00 2012.

\bibitem{MR2585987}
Thomas Chen, J{\"u}rg Fr{\"o}hlich, and Alessandro Pizzo.
\newblock Infraparticle scattering states in non-relativistic {QED}. {I}. {T}he
  {B}loch-{N}ordsieck paradigm.
\newblock {\em Comm. Math. Phys.}, 294(3):761--825, 2010.

\bibitem{MR2238912}
Jan~Philip Solovej.
\newblock Upper bounds to the ground state energies of the one- and
  two-component charged {B}ose gases.
\newblock {\em Comm. Math. Phys.}, 266(3):797--818, 2006.

\bibitem{MR1347692}
Eberhard Zeidler.
\newblock {\em Applied functional analysis}, volume 109 of {\em Applied
  Mathematical Sciences}.
\newblock Springer-Verlag, New York, 1995.
\newblock Main principles and their applications.

\end{thebibliography}
